\documentclass{amsart}

\usepackage{amssymb}
\usepackage{graphicx}
\usepackage{graphics}

\def\One{\mathbb{I}}

\newtheorem{theorem}{Theorem}[section]
\newtheorem{lemma}[theorem]{Lemma}

\theoremstyle{definition}
\newtheorem{definition}[theorem]{Definition}

\theoremstyle{remark}

\numberwithin{equation}{section}

\begin{document}

\title[Landau poles and cut structure by Hopf algebra]{Some results on Landau poles and Feynman diagram cut structure by Hopf algebra}

\author{William Dallaway}
%\address{}
%\curraddr{}
%\email{}
\thanks{}

%    author two information
\author{Karen Yeats}
%\address{}
%\curraddr{}
%\email{}
\thanks{KY is supported by an NSERC Discovery grant and the Canada Research Chairs program.  KY thanks Dirk Kreimer for longstanding collaborations that led to this paper.}

\date{}

\begin{abstract}
We investigate a system of differential equations for the beta function of massless scalar $\phi^4$ theory and continue the combinatorial investigation of the cut structure of Feynman diagrams.
\end{abstract}

\maketitle

\section{\label{sec:level1} Introduction}

The goal of this work is to analyse the anomalous dimension and the beta function of a scalar quantum field theory using a differential equation obtained through combinatorial techniques. We also consider combinatorial questions related to the coaction of the cut structure of Feynman diagrams.

\medskip

Broadhurst and Kreimer in \cite{bkerfc} derived a non-linear first order differential equation for the anomalous dimension of a particular approximation to the fermion propagator in a massless Yukawa theory and a non-linear third order differential equation for the scalar $\phi^3$ case.  In \cite{YeatsThesis} one of us developed a more general approach following the original Broadhurst and Kreimer Yukawa example that quite broadly builds a non-linear first order differential equation or system of non-linear first order differential equation for anomalous dimensions of quantum field theories.  The cost of this level of generality was that much of the complexity and character of the specific quantum field theory is swept into a catch-all function $P$, but a benefit is that the overall shape of the differential equation is controlled by combinatorial features of the insertion of Feynman diagrams in the theory.  The specific cases of the photon propagator of QED in a Baker-Johnson-Wiley gauge and the gluon propagator in similarly special gauge in massless QCD were studied in \cite{van_Baalen_2009} and \cite{vBKUY2} respectively.  These both correspond to single differential equations in this framework.  The massless $\phi^4$ case, a system of two differential equations, was set up at the end of \cite{YeatsThesis} but never subsequently investigated until the present paper.  

The differential equation set up is particularly well suited to understanding when there are solutions for the anomalous dimension that exist for all values of the coupling and when there are solutions that exist for all values of the energy scale, that is, when there are or are not Landau poles.  Consequently, these questions are a focus in the case studied here as well as in the previously studied cases of \cite{van_Baalen_2009} and \cite{vBKUY2}.

Finally, let us note that the original two differential equations of Broadhurst and Kreimer have recently been very strikingly studied from a resurgence perspective in \cite{BORINSKY2020115096, BDMphi3, BORINSKY2022115861}, and Bellon and collaborators investigated the anomalous dimension in the Wess-Zumino model \cite{bel2008, BShigher, BCorder} in a similar context.

\medskip

An algebraic and combinatorial approach to the relationship between the cut structure and subdivergence structure of Feynman diagrams was recently developed by one of us along with Dirk Kreimer \cite{Algebraic}.  The cut structure of a Feynman diagram controls its monodromy around singularities and is related to its infrared behaviour while the subdivergence structure controls its behaviour under renormalization.  However, the set up of \cite{Algebraic} required fixing a spanning tree, as motivated by connections to Cullen and Vogtmann's Outerspace \cite{CVouterspace, BeK, Kouterspace}, something which is unnecessary and at times cumbersome.  We rectify this here with a tree-independent formulation and then look at some of the combinatorial preliminaries necessary for extending Klann's calculations \cite{KlannMSc} on how the cut structure relates to the infrared divergences.

\medskip

The basis of all of this are the renormalization Hopf algebras first developed by Connes and Kreimer and first we will give a brief account of this. We then present the differential equations which will be the main focus of the analysis. The remainder of the introduction considers extensions of the Hopf algebras and a coaction studied in \cite{Algebraic}.  

For a mathematical reader uninterested in the physics background, the key things are the differential equations \eqref{eq de1}, \eqref{eq de2} along with a general sense that there is some motivation to studying them and especially to studying when solutions exist for all $x$ and all $L$, and the core and cut coproducts on graphs Definition~\ref{def core} and \eqref{eq modified core coprod}, and \eqref{eq cut coprod}, along with a general sense that the relationship between them is meaningful.  Everything else in this introduction is motivation.

\subsection{Feynman diagrams and renormalization Hopf algebras}

Quantum field theory is a framework in which we can understand arbitrary numbers of interacting particles quantum mechanically.  It is the standard way to unify quantum mechanics and special relativity.  In high energy physics, a prototypical experiment consists of generating known particles, colliding them together at high energy, and investigating the particles that appear as a result.  On the theoretical and mathematical side of high energy physics, then, we want to better understand the mathematical structures in quantum field theory so as to be able to better calculate the probability amplitudes of particle interactions as occur in such a collider experiment.

One longstanding but still important technique for computing these amplitudes is the loop expansion in Feynman diagram.  In this approach amplitudes are computed and studied as given by an infinite series indexed by Feynman diagrams.  Very roughly Feynman diagrams are graphs where the edges represent particles propagating and the vertices represent particle interactions.  Each graph contributes to the amplitude via its Feynman integral, an integral which can be read off of the graph.  Particles entering or exiting are represented as unpaired half-edges in the graph.  The \emph{loop number} of the graph is the dimension of the cycle space of the graph, and in the loop expansion of an amplitude, the series is organized by increasing number of loops.

Formally we will define graphs as follows.  

\begin{definition}
A \emph{graph} $G$ consists of
\begin{itemize}
    \item a set $H(G)$ of the \emph{half edges} of the graph,
    \item a set partition, $\mathcal{V}(G)$, of $H(G)$ into parts of size at least 3, the parts of which are the \emph{vertices} of the graph, and
    \item a set partition, $\mathcal{E}(G)$, of $H(G)$ into parts of size at most 2, the parts of which are the \emph{edges} of the graph.
\end{itemize}
\end{definition}
This notion of graph is a little different from what is usual in graph theory.  It allows multiple edges and self-loops, but does not allow vertices of degree less than 3 -- in this formulation the \emph{degree} of a vertex is its size as a part in the set partition.  Additionally, $\mathcal{E}(V)$ may include edges where the size of the part is 2 which we call \emph{internal edges} and which correspond to edges in the usual graph theory sense, but also edges where the size of the part is 1 which we call \emph{external edges} and which correspond, as described roughly above, to the particles entering and exiting the system.  This notion of graph is closely related to the usual formulation of a \emph{combinatorial map} (see for instance \cite{LZgraphs}) but without any cyclic ordering of the half edges at each vertex.

Despite these differences most graph theory notions can be inherited directly from their usual versions (see \cite{Dbook} for a standard graph theory introduction), and so we speak of notions such as connectivity for our graphs without further definition.

To move from these graphs to our formulation of Feynman diagrams, we will extract only the algebraic or combinatorial part that we need from the quantum field theory in the following definition.

\begin{definition}
A \emph{combinatorial physical theory} is a set of half edge types along with
\begin{itemize}
    \item a set of pairs of not necessarily distinct half edge types defining the permissible edge types,
    \item a set of multisets of half edge types defining the permissible vertex types,
    \item an integer called the \emph{power counting weight} for each edge type and each vertex type, and
    \item a nonnegative integer \emph{dimension of spacetime}.
\end{itemize}
\end{definition}
Typically, the set of half edge types will be finite.  The dimension of spacetime should be the one at which the theory is renormalizable but not superrenormalizable.  The combinatorial meaning of this will be given later.  

For example, here are some standard quantum field theories in this framework.
\begin{itemize}
    \item   Quantum electrodynamics (QED) has 3 half edge types, a half photon, a front half fermion, and a back half fermion. There are two edges type, the pair of two half photons, giving a photon edge which is drawn as a wiggly line and which has power counting weight 2, and the pair of a front half fermion and a back half fermion, giving a fermion edge, which by its construction is oriented and drawn as an edge with an arrow, and has power counting weight 1. There is one vertex consisting of one of each half edge type and with weight 0. The dimension of spacetime is 4.
    \item Scalar $\phi^4$ theory has just one half edge type and just one edge type consisting of a pair of the half edges drawn as a plain unoriented edge.  This edge type has weight 2. The $4$ of $\phi^4$ indicates that the one vertex is a 4-valent vertex; that is, the vertex consists of a multiset of 4 copies of the half edge. It has weight 0.  The dimension of spacetime is 4.
\end{itemize}
For further examples see \cite{YeatsThesis}.

\begin{definition}
A \emph{Feynman graph} in a given combinatorial physical theory is a graph in the sense above along with a map from the half edges of the graph to the set of half edge types such that every internal edge is of a permissible edge type and every vertex is of a permissible vertex type.
\end{definition}

The step from this point to actual loop expansions in quantum field theory is a map known as the \emph{Feynman rules} which associate a Feynman integral to each Feynman graph.  There are many different takes on the Feynman rules; for a standard physics take consider a standard quantum field theory text book such as \cite{iz} or \cite{psbook}, while for a more mathematical viewpoint one might see the Feynman rules as coming from the exponential map on the Lie algebra associated to the renromalization Hopf algebra \cite{Lskript}, but in any case the details won't be important for us.

However, a property of the Feynman integrals that will matter is that in most interesting cases they are divergent integrals.  Because of this it is best to think of the Feynman rules as mapping not to integrals but to formal integral expressions \cite{YeatsThesis, Pphd} so the integrands or differential forms can be manipulated and compared.  Divergences of a Feynman integral come in a few types.  There may be a proper subgraph that already corresponds to a divergent integral; this is a \emph{subdivergence}.  Divergences may come when momenta are taken in a limit where they are large, these are \emph{ultraviolet} (UV) divergences; or when some aspect of the story becomes small, for example when two particles become parallel so the angle between them goes to $0$, these are \emph{infrared} (IR) divergences.  The IR divergences are much more subtle and we will postpone any discussion of them until we have the machinery of Cutkosky cuts in place.  The UV divergences are more straightforward.  The power counting weights in the definition of combinatorial physical theory are the power to which the factor given by the corresponding edge or vertex grows as the momenta get large.  Because of this we can determine the overall degree of UV divergence of a Feynman graph from only these combinatorial considerations.

\begin{definition}
For a Feynman graph G in a combinatorial physical theory, let $w(a)$ be the power counting weight of an internal edge or a vertex  $a$ of $G$ and let $D$ be the dimension of spacetime. Then the \emph{superficial degree of divergence} is
\[
D\ell - \sum_{e \in \mathcal{E}(G)} w(e) - \sum_{v\in \mathcal{V}(G)} w(v)
\]
where $\ell$ is the loop number of the graph.

If the superficial degree of divergence of a graph is nonnegative we say the graph
is (UV) \emph{divergent}. If it is 0 we say the graph is (UV) \emph{logarithmically} divergent.
\end{definition}
Note that as $D$ increases, more graphs are deemed divergent.  For the theories of interest to use there is a special value of $D$ where the superficial degree of divergence of a connected graph depends only on its external edges.  This is the value of $D$ that corresponds to the physical situation of the theory being renormalizable but not superrenormalizable, and is the value of $D$ we typically want in our combinatorial physical theories.

\medskip

The problem of how to deal with UV divergences has long been understood by physicists, and beautifully, one perspective on it can be reformulated using a combinatorial Hopf algebra.

At this point it is also important to note what notion of subgraph we want.  
\begin{definition}
A \emph{subgraph} $\gamma$ of a graph $G$ is a graph where $H(\gamma)\subseteq H(G)$, $\mathcal{E}(\gamma)$ is a refinement of $\mathcal{E}(G)$ restricted to $H(\gamma)$ and $\mathcal{V}(\gamma)$ is exactly $\mathcal{V}(G)$ restricted to $H(\gamma)$ with the additional property that every part in $\mathcal{V}(\gamma)$ has the same size in $\mathcal{V}(G)$. 
\end{definition} 
For example, consider the graph on the left in Figure~\ref{fig:bananas}.  It has the graph on the right in the figure as a subgraph in three ways.  In each case the set of half edges in the subgraph is the same as in the original graph, and the set of vertices is thus necessarily also the same, but the set of edges in each subgraph refines exactly one of the internal edges of the original graphs by breaking the part of size two corresponding to the internal edge into two parts of size 1, giving two new external edges.  
\begin{figure}
    \centering
    \includegraphics{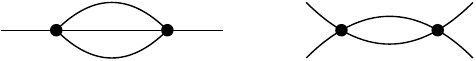}
    \caption{The right hand graph is a subgraph of the left hand graph in three ways.}
    \label{fig:bananas}
\end{figure}

Note that this definition of subgraph guarantees that if a graph is a Feynman graph then the subgraph will also be a Feynman graph with the inherited mapping to the half edge types.

We define the contraction of a connected subgraph $\gamma$ within a graph $G$ differently depending on whether $\gamma$ has more than two external edges or not.  If $\gamma$ has more than two external edges, then simply contract all internal edges of $\gamma$, so $\gamma$ becomes a new vertex made of its external edges.  If $\gamma$ has two external edges, then to avoid a 2-valent vertex, remove both internal and external edges of $\gamma$ and then join into an edge the two half-edges originally joined to the two external edges of $\gamma$, or leave as external, the half if there was only one.  Subgraphs $\gamma$ with one external edge will not come up for us, but they would be treated also by removing all edges.  For a disconnected subgraph $\gamma$ contract it by contracting each connected component.

As is often done in enumerative combinatorics as well as in quantum field theory we can reduce to considering connected Feynman graphs by the exp-log transformation, and furthermore, using a Legendre transform \cite{JKMtowards} we can restrict to considering \emph{one particle irreducible} (1PI), that is bridgeless in the graph theory sense.  Now we are ready to define a Hopf algebra structure on the connected 1PI Feynman diagrams of a combinatorial physical theory.

For the definition of the renormalization Hopf algebra we need one more property for the combinatorial physical theory, we need that the external edges of any divergent subgraph appear as a vertex or edge type of the theory.  If this is not the case, extend the allowable vertex and edge types as necessary.
\begin{definition}
Given a combinatorial physical theory as above, let $\mathcal{G}$ be the set of connected 1PI Feynman graphs in that theory.  Define the \emph{renormalization bialgebra}, $\mathcal{H}$, of the theory as follows.  As an algebra $\mathcal{H} = \mathbb{Q}[\mathcal{G}]$ and we identify disconnected graphs with the monomial of their connected components.  The coproduct is defined on elements of $\mathcal{G}$ by
\[
    \Delta(G) = \sum_{\substack{\gamma \subseteq G \\ \gamma \text{  divergent 1PI}}} \gamma \otimes G/\gamma
\]
and extended as an algebra homomorphism to $\mathcal{H}$.  The counit is defined by $\epsilon(1)=1$ and $\epsilon(G)=0$ for $G$ with loop order at least 1 and extended as an algebra homomorphism.  
(Note that the subgraphs in the sum do not need to be connected).
\end{definition}

This turns out to be a bialgebra.  Furthermore, $\mathcal{H}$ is graded by the loop order of the graphs, and so by taking a quotient which identifies all zero loop graphs in $\mathcal{G}$ with 1 we get a graded connected bialgebra and so automatically a Hopf algebra, the \emph{renormalization Hopf algebra} of the combinatorial physical theory.  Note that since $\mathcal{G}$ consists of connected 1PI graphs, the only possible zero loop graphs are single vertices.  See \cite{Bphd} for a recent presentation of these structures originally due to Kreimer \cite{Khopf} and Connes and Kreimer \cite{ckI}.

The \emph{antipode} is the extra map that makes a bialgebra a Hopf algebra, and in this context it is of key importance, since it tells us recursively how to break a Feynman graph into subgraphs so as to take care of the UV subdivergences.  To do so takes an additional step since the antipode is a map from the Hopf algebra to itself, but to renormalize we need a map to the integrals.  However, this renormalization map has the same recursive form as the antipode with a regularization scheme and the Feynman rules applied in the appropriate places \cite{ckI} (see also, eg, section 2.3.2 of \cite{YeatsThesis}).  For us we need only understand that the Hopf algebra (of which the coproduct is the only nontrivial map) controls renormalization, and so any algebraic or combinatorial structure that is compatible with the coproduct will  likewise be compatible with renormalization, which is important from a physical point of view.

We can also define the core Hopf algebra on the set of all 1PI graphs.
\begin{definition}\label{def core}
Lt $\mathcal{G}$ be the set of connected 1PI graphs.  Define the \emph{core bialgebra}, $\mathcal{H}_{\text{core}}$, as follows.  As an algebra $\mathcal{H} = \mathbb{Q}[\mathcal{G}]$ and we again identify disconnected graphs with the monomial of their connected components.  The coproduct is defined on elements of $\mathcal{G}$ by
\[
    \Delta(G) = \sum_{\substack{\gamma \subseteq G \\ \gamma \text{  1PI}}} \gamma \otimes G/\gamma
\]
and extended as an algebra homomorphism to $\mathcal{H}$.  The counit is defined by $\epsilon(1)=1$ and $\epsilon(G)=0$ for $G$ with loop order at least 1 and extended as an algebra homomorphism. 
\end{definition}
The core bialgebra can be made into the \emph{core Hopf algebra} as for renormalization Hopf algebras.  Note that the core Hopf algebra can be thought of as the limit, when the dimension of spacetime grows, of the renormalization Hopf algebra for a physical theory with one undirected edge type, since eventually every 1PI graph becomes divergent, and every vertex degree is required.

\subsection{The anomalous dimension}

Thinking back to our idealized particle scattering experiment, we were not really interested in the probability amplitude of any specific Feynman integral, but rather we used the loop expansion in Feynman integrals to understand the amplitude of some particular process.  In this context that means that we fix the external edges and then are interested in the series indexed by Feynman diagrams with those external edges where each diagram contributes its Feynman integral.  This is a series in the \emph{coupling constant}, a parameter measuring the strength of the particle interaction, which we will call $x$. Using Euler's formula and properties of the vertices of the theories of interest to us, we see that the power of the coupling constant for a Feynman graph in a combinatorial physical theory of interest to us grows with the loop order and so we can think of this series as a kind of generating function in the sense of enumerative combinatorics which counts Feynman graphs by loop order and weighted by their Feynman integral.  

In the cases of interest these series are expected to be divergent series (the number of Feynman diagrams grows factorially, but the Feynman integrals can come in either sign and while there are many indications that the series should still diverge, this is mostly not proven).  For the purposes of this paper, it suffices to think of these series simply as formal power series, as one might expect from the enumerative perspective.  However, from an analytic perspective they are asymptotic series, can typically be usefully resummed, and have interesting resurgence properties a topic of significant recent interest, see for example \cite{BCalien, BORINSKY2020115096, BDMphi3, BORINSKY2022115861}, but this will not be important for the present purposes.

As mentioned above, these series can be reduced to series indexed by connected 1PI diagrams by log and Legendre transforms and so the key series for us will be series over connected 1PI diagrams with a fixed multiset of external edges, either over all such diagrams or over particular subclasses of them.  These series we will call the \emph{Green functions}.  As well as being a function of the coupling constant, the Green function is a function of the kinematical parameters: the momenta of the (fixed) incoming and outgoing particles.  For our purposes we will restrict to the single scale case where there is only one such kinematical parameter.  We will let $L$ denote the log of the norm of this one momentum going through the graph and our Green functions we will consider as formal series in the variables $x$ and $L$.  We will denote these Green functions as $G^{(r)}(x,L)$, where the index $r$ corresponds to the multiset of external edges.

There is a very important equation satisfied by the Green functions, known as the \emph{renormalization group equation} or \emph{Callan–Symanzik equation}.  One might have hoped that the coupling constant had a small value since we are expanding a series in it, but one consequence of renormalization is that the coupling constant depends on the energy level, that is on $L$, a phenomenon known as the \emph{running} of the coupling.  
The renormalization group equation captures how change in $x$ and change in $L$ affect $G(x, L)$. Specifically
\[
\left(\frac{\partial}{\partial L} + \beta(x)\frac{\partial}{\partial x} - \gamma^{(r)}(x)\right)G^{(r)}(x,L) = 0
\]
The $\beta(x)$ physically encodes the flow of the coupling depending on the energy scale and $\gamma^{(r)}(x)$ is the \emph{anomalous dimension} of $G^{(r)}(x,L)$ and as a formal series is the coefficient of $L^1$ in $G(x,L)$.  When we have a system of Green functions built from inserting graphs with different $r$ into each other, as we'll see in Dyson-Schiwnger equations shortly, then the $\beta(x)$ is a linear combination of the various $\gamma^{(r)}(x)$ weighted by some combinatorial factors.  See Section 4.1 of \cite{YeatsThesis} for a discussion.

\medskip

Another important equation for the Green function is the Dyson-Schwinger equation.  Dyson-Schwinger equations are the quantum analogues of the equations of motion.  At the level of the Feynman graphs they are built by inserting graphs into other graphs recursively and so they behave very much like functional equations for generating series in algebraic combinatorics.  See \cite{YeatsThesis} for examples.  After applying Feynman rules, the Dyson-Schwinger equations are integral equations for the Green functions.   The explicit integral gives the contribution of the graph into which we are inserting while the appearance of the Green function inside the integral gives the recursive contribution of the graphs being inserted.  The outermost graph being inserted into should be primitive in the renormalization Hopf algebra, that is it should be itself divergent but have no proper subdivergences.  We can insert other graphs into vertices and edges of this outer graph.  
To build all connected 1PI graphs with a given multiset of external edges, we would expect to need to insert into all the edges and vertices, and may need infinitely many primitives to insert into.

One typically has a system of Dyson-Schwinger equations for $G^{(r)}(x,L)$ for different $r$s.  Sometimes we restrict our attention to building graphs only out of certain pieces or work in particular theories where we can restrict to a single equation.  We also make some other simplifications.  We have been writing our Green functions with only a single kinematical parameter $L$, but vertex insertions will have multiple kinematical parameters since they have more than two external edges.  We get around this by considering one overall energy scale and treating the rest as constant angles that we simply carry along, see \cite{KRUGER2015293}. 

As a concrete example Broadhurst and Kreimer studied two particular examples of Dyson-Schwinger equations in \cite{bkerfc}. 
For their Yukawa theory example, the Dyson-Schwinger equation (in a different normalization than theirs) is 
\[
G(x,L) = 1 - \frac{x}{q^2}\int d^4x \frac{k\cdot q}{k^2 G(x, \log(k^2/\mu^2)) (k+q)^2} - \text{same integrand}\bigg|_{q^2=\mu^2}
\]
where $L=\log(q^2/\mu^2)$.  We can convert this into a pseudo-differential form, see p72 of \cite{Ybook},
\[
G(x,L) = 1-xG(x, d/d(-\rho))^{-1}(e^{-L\rho}-1)F(\rho)|_{\rho=0}
\]
where $F(\rho)$ is the Feynman integral of the primitive regularized by raising the propagator we insert into to the power $1+\rho$; in this particular case it is
\[
F(\rho)= \int d^4x \frac{k\cdot q}{(k^2)^{1+\rho}(k+q)^2} \bigg|_{q^2=1}.
\]
One nice thing is that the $\rho$ in $F(\rho)$ is functioning as a regulator, but it was not added in by hand.  It appeared naturally over the course of the calculation.

All the Dyson-Schwinger equations that are of interest to us can be put into such a pseudo-differential form.  The most general form considered in \cite{YeatsThesis} is 
\[
G^{(r)}(x,L) = 1-\text{sgn}(s_r)\sum_{k\geq 1} x^k G^{(r)}(x, \partial/\partial \rho) \prod_{j} (G^{(j)}(x, \partial/\partial \rho))^{s_j} (e^{-L\rho} -1) F_{k,r}(\rho) \bigg|_{\rho = 0}
\]
where the product is over the multisets of external edges which we consider in our theory; typically those which give divergent Feynman graphs.
In this formula, the $s_j$ are integers that keep track of how the number of insertion places for the multiset of external edges $j$.  There is a nice combinatorics of insertion counting, see Section 3.3.2 of \cite{YeatsThesis} or Section 5.4 of \cite{Ybook}.
Additionally, as well as simplifying by working with only a single scale $L$, there is also a simplification by using a symmetric insertion as in Section 2.3.3 of \cite{YeatsThesis} to only keep track of one insertion place, that is each $F_{k,r}$ has only a single $\rho$ as its argument. 

The pseudo-differential form has been very useful, notably as the basis for the chord diagram expansion solutions of Dyson-Schwinger equations \cite{MYchord, HYchord, CYchord, CYZchord, CYnnl}.  For the present purposes, however, we want to further reformulate our Dyson-Schwinger equations.  Write $G^{(r)}(x,L) = 1-\text{sgn}(s_r)\sum_{k\geq 1} \gamma^{(r)}_k(x)L^k$.  Then the anomalous dimension is $\gamma^{(r)}(x) = -\text{sgn}(s_r)\gamma^{(r)}_1(x)$.  Furthermore, extracting the coefficient of $L^k$ from the renormalization group equation and using the expression for $\beta(x)$ in terms of the anomalous dimensions we get
\[
\gamma^{(r)}_k(x) = \frac{1}{k} \left(\text{sgn}(s_r)\gamma^{(r)}_1(x) -\sum_j |s_j|\gamma^{(j)}_1(x) \frac{x\partial}{\partial x}\right) \gamma^{(r)}_{k-1}(x)
\]
see Theorem 4.1 of \cite{YeatsThesis} for details.

Now, if each $F_{k,r}(\rho)$ happen to be of the form $a_{k,r}/(\rho(1-\rho))$, as the unique $F(\rho)$ is for the Broadhurst-Kreimer Yukawa example, then something special happens.  Taking the coefficient of $L$ and $L^2$ in the pseudo-differential form of the Dyson-Schwinger equation, we obtain
\begin{align*}
2\gamma^{(r)}_2(x) 
& = -\text{sgn}(s_r)\sum_{k\geq 1} x^k G^{(r)}(x, \partial/\partial \rho) \prod_{j} (G^{(j)}(x, \partial/\partial \rho))^{s_j} \frac{a_{k,r}\rho}{1-\rho} \bigg|_{\rho = 0} \\
& = -\text{sgn}(s_r)\sum_{k\geq 1} x^k G^{(r)}(x, \partial/\partial \rho) \prod_{j} (G^{(j)}(x, \partial/\partial \rho))^{s_j} \left(\frac{a_{k,r}}{1-\rho} - a_{k,r}\right) \bigg|_{\rho = 0} \\
& = \gamma^{(r)}_1(x) + \sum_{k\geq 1}a_{k,r}x^k.
\end{align*}
The recurrence we extracted from the renormalization group equation also gives a relationship between $\gamma_2$ and $\gamma_1$ so combining them we get
\[
\gamma^{(r)}_1(x) = P^{(r)}(x) - \text{sgn}(s_r)\gamma^{(r)}_1(x)^2 + \sum_j |s_j|\gamma^{(j)}_1(x)x \frac{d}{dx} \gamma^{(r)}_1(x)
\]
where $P^{(r)}(x) = \sum_{k \geq 1}a_{k,r}x^k$.  Physically what this $P^{(r)}$ series is doing is taking the leading behaviour of each primitive graph and combining those into a series.  The further observation of Chapter 7 of \cite{YeatsThesis} is that even if the $F_{k,r}(\rho)$ are not geometric series, then we can still get a differential equation for $\gamma^{(r)}_1(x)$ of the same form at the cost of $P^{(r)}(x)$ containing not just the leading behaviour of each primitive but some extra stuff that's simply what it needs to be to make the formula hold.  This means we lose much of our physical intuition for $P^{(r)}(x)$.

None the less, this form of differential equation for the anomalous dimension was usefully studied for single equation versions of QED and QCD  and related cases in \cite{van_Baalen_2009, vBKUY2}, as well as being used in the Borinsky and Dunne's resurgent analysis \cite{BORINSKY2020115096} of the Broadhurst-Kreimer Yukawa case \cite{bkerfc} where $P(x)=x$ so there is no loss of physical meaning.

Near the end of \cite{YeatsThesis}, it was suggested that an interesting next case would be the two equation system for the vertex and propagator in massless $\phi^4$ theory:
\begin{align}\label{eq de1}
\gamma^+(x) & = P^+(x) + \gamma^+(x)^2 + (\gamma^+(x) + 2\gamma^-(x))x\frac{d}{dx}\gamma^+(x) \\
\label{eq de2}
\gamma^-(x) & = P^-(x) + \gamma^-(x)^2 + (\gamma^+(x) + 2\gamma^-(x))x\frac{d}{dx}\gamma^-(x),
\end{align}
but until now, this was not followed up on.  In this equation the $+$ indicates the vertex and the $-$ indicates the propagator.  $x(\gamma^+(x)+2\gamma^-(x))$ is the $\beta$-function and the dependence of $x$ on $L$ is as usual given by $\frac{dx}{dL} = \beta(x(L))$.

Even with the uncertainties in the interpretation of $P^{(r)}(x)$, one thing the $\gamma_1$ differential equation can be quite useful for, as was studied in \cite{van_Baalen_2009, vBKUY2}, is determining what conditions on $P^{(r)}(x)$ give a $\beta(x)$ which exists for all values of the coupling, or exists for all values of $L$.  If $\beta(x)$ diverges in finite $L$ then we say that we have a Landau pole.  Theories that are not asymptotically free are typically expected to have a Landau pole, but arguments from perturbation theory for Landau poles are not necessarily convincing since an asymptotic series doesn't necessarily say much away from the point of expansion.

\subsection{Cut Hopf algebra and coaction}\label{subsec orig coact}

There are more mathematically interesting complexities in Feynman graphs and Feynman integrals.  Results due to Cutkosky \cite{Ccut} tell us that the behaviour of a Feynman integral around a singularity can be computed from calculating Feynman integrals for graphs with cuts.  See the introduction of \cite{BlKcut} for a mathematical summary.  Cuts of the graph into two pieces are not sufficient; cuts into more pieces are needed to understand so called anomalous thresholds \cite{Canomalous}.  Cuts are also related to infrared singularities via the sector decomposition, see \cite{Algebraic} section 4.3 for a discussion.

An algebraic approach to cuts is via a cut coaction or coproduct \cite{Kreimer2021}.  A principal point of \cite{Algebraic} is to understand how the cut coproduct and the core coproduct interact.  We will summarize this construction here.

In order to bring cuts into the Hopf algebraic framework as in \cite{Algebraic}, we will rephrase the core Hopf algebra in terms of fixing a spanning tree and then taking certain sets of edges not in the tree.  This makes it easier to have the cut structure on the same level, since with a fixed spanning tree cuts are determined by sets of edges in the tree.

Following \cite{Algebraic} Appendix C, for this subsection we will work with pairs $(G,T)$ of a graph $G$ and a spanning tree $T$ of $G$.  We can take the span of such pairs in order to form a vector space.  We view $G$ itself in this context as the sum over all of its spanning trees $\sum_{T}(G, T)$ and in that way embed the vector space of graphs into the vector space of graph tree pairs.

Given a pair $(G,T)$ and an edge $e \in E(G)-E(T)$, the graph $T\cup e$ has exactly one cycle.  This cycle is known as the \emph{fundamental cycle} determined by $e$ and $T$.  For a fixed spanning tree, the fundamental cycles give a basis for the cycle space of the graph.  Running over all spanning trees, all cycles of the graph appear (typically multiple times) as fundamental cycles.

Also, given a pair $(G, T)$ and an edge $e\in E(T)$, the graph $T-e$ has exactly two components, so it partitions the vertices of $G$ into two parts.  The edges of $G$ with one end in each part define a cut of $G$ into two components; $e$ is one of the edges of the cut, though typically there will be more.  Removing $k$ edges of $T$ partitions the vertices of $G$ into $k+1$ parts.  The edges with ends in different parts again define a cut of $G$, this time a cut into $k+1$ components.

\medskip

Using the fundamental cycles we can rephrase the core Hopf algebra.  A subset $S\subseteq E(G)-E(T)$ determines a set of fundamental cycles and taking the union of these cycles within $G$ we obtain a $1PI$ subgraph (not necessarily connected) $\gamma_S$.  We can also perform a complimentary operation, given a subset $S\subseteq E(G)-E(T)$, contract all fundamental cycles other than those associated to edges in $S$.  Call the resulting graph $\gamma^{S}$.  Define
\[
\Delta(G,T) = \sum_{S\subseteq E(G)-E(T)} (\gamma_S, T\cap \gamma_S) \otimes (\gamma^{E(G)-E(T)-S}, T\cap \gamma^{E(G)-E(T)-S})
\]
This defines a coproduct on graph tree pairs.  Ignoring the spanning trees, all these terms appear in the core coproduct of $G$ and summing over spanning trees we obtain the core coproduct itself, using the embedding described above, see Section 3.5 of \cite{Algebraic}.
Additionally, note that whether to take $\gamma_S$ or $\gamma^S$ can determined purely positionally based on which side of the tensor one is on.  Even for an iteration of the coproduct the position determines what to do with a subset; things coming to the left are contracted, and the remaining subgraph is taken.  Consequently, we can ignore the graphs entirely and see that the subset coproduct
\[
\Delta(S) = \sum_{S'\subseteq S} S' \otimes (S-S')
\]
applied to $E(G)-E(T)$ carries the same information as the graph tree pair coproduct.

\medskip

We can do a similar thing to define a cut coproduct.  It will be useful not only to indicate edges to define the cut, but also to be able to indicate edges that are uncuttable, or equivalently edges which are contracted.  To do this we will use an interval in the power set of $E(T)$.  Write $[A, B] = \{S\subseteq E(T): A\subseteq S \subseteq B\}$ and interpret $[A,B]$ as $A$ defining a cut as above and $E(T)-B$ defining the uncuttable edges.  Then the coproduct we want is the incidence coproduct on these intervals
\[
\rho([A,B]) = \sum_{A\subseteq S\subseteq B} [A,S]\otimes[S,B]
\]
which we interpret as a coproduct on graph tree pairs with tree edges marked for cutting or as uncuttable, see Section C.2 of \cite{Algebraic}.

\medskip

To see how these two coproducts fit together we need to put them on a common footing, so fix $(G,T)$, write $\mathcal{P}$ for the power set operator, define $f:E(G)-E(T) \rightarrow \mathcal{P}(E_T)$ to be the fundamental cycle map, and define $\mathcal{A}$ to be the span of the formal symbols $B^{[A_1, A_2]}$ where $B\subseteq E(G)-E(T)$ and $[A_1, A_2]$ is an interval in $\mathcal{P}(f(B))$.  Then $\Delta$ and $\rho$ act on this space of symbols by acting on the $B$ part and the $[A_1, A_2]$ part respectively and carrying along the other part, restricted as necessary.

We may additionally add restrictions on the interval $[A_1, A_2]$ by having certain edges which are forbidden to become tadpoles (such as fermion edges which give zero in the Feynman integral if they appear as tadpoles).  The reader can consult Sections C.2 and C.3 of \cite{Algebraic} if they wish to see this detail worked out.

It remains to consider the product.  Since subgraphs are represented by unions of fundamental cycles within $G$, the product for these bialgebras must also be union within $G$, with the result being $0$ if the intervals are not compable with taking such a union.  This is a bit intricate to write out but we obtain the following.
The product is
\[
m(B_1^{[A_1, A_2]}, B_2^{[A_3, A_4]}) =
\begin{cases}
  (B_1\cup B_2)^{[A_1\cup A_3, A_2\cup A_4]} & \text{if } B_1\cap B_2=\emptyset, A_1|_{A'} = A_3|_{A'}\\
    & \quad \text{ and } A_2|_{A'} = A_4|_{A'}; \\                              
0 & \text{otherwise,}                                                           
\end{cases}                                                                     
\]
where $A'=f(B_1)\cap f(B_2)$ and $A|_{A'}$ is the set $A$ restricted to the set $A'$, that is $A\cap A'$.
The unit for this product is $\One = \emptyset^{[\emptyset, \emptyset]}$.
The core-type coproduct is
\begin{equation}\label{eq modified core coprod}                                                                    
\Delta(B^{[A_1, A_2]}) = \sum_{\substack{B_1\subseteq B \\ f(B_1)\cap A_1 = \emptyset}} B_1^{[A_1\cap f(B_1), A_2\cap f(B_2)]} \otimes (B\backslash B_1)^{[A_1\cap f(B\backslash B_1), A_2\cap f(B\backslash B_1)]}                         
\end{equation}
The counit for this coproduct is $\hat{\One}_{\Delta}(\One) = 1$ and $\hat{\One}_{\Delta}(B^{[A_1, A_2]}) = 0$ for $B\neq \emptyset$.
The cut coproduct is
\begin{equation}\label{eq cut coprod}
\rho(B^{[A_1, A_2]}) = \sum_{A_1\subseteq A \subseteq A_2} B^{[A_1, A]}\otimes B^{[A, A_2]}
\end{equation}
Let $\mathcal{A}_p$ be the subspace of $\mathcal{A}$ spanned by $B^{[\emptyset, A]}$, then $(\mathcal{A}_p, m, \Delta)$ is a Hopf algebra and $\Delta$ gives a coaction of $\mathcal{A}$ on $\mathcal{A}_p$, $\Delta: \mathcal{A} \rightarrow \mathcal{A}_p \otimes \mathcal{A}$, see \cite{Algebraic} discussion at the beginning of Section C.3.  Furthermore (\cite{Algebraic} Lemma C.1) $(\mathcal{A}, m, \rho)$ is a bialgebra.  Most importantly, with the coaction $\rho$, $(\mathcal{A}_p, m,\Delta)$ and
$(\mathcal{A},m,\rho)$ are in cointeraction (\cite{Algebraic} Theorem C.2).  Concretely this means that
on $\mathcal{A}$
\begin{itemize}
\item $\rho(\One)=\One\otimes\One$.
\item $m_{1,3,24}\circ(\rho\otimes\rho)\circ \Delta=(\Delta\otimes\mathrm{id})\circ\rho$, with:
\[                                                                              
m_{1,3,24}:\,\mathcal{A}\otimes \mathcal{A}\otimes \mathcal{A}\otimes \mathcal{\
A}\to \mathcal{A}\otimes \mathcal{A}\otimes \mathcal{A},                        
\]
\[                                                                              
m_{1,3,24}(w_1,w_2,w_3,w_4)=w_1\otimes w_3\otimes (w_2 w_4).                    
\]
\item $\forall w_1,w_2\in \mathcal{A}$, $\rho(w_1 w_2)=\rho(w_1)\rho(w_2)$,
\item $\forall w\in\mathcal{A}$, $(\hat{\One}_{\Delta}\otimes \mathrm{id})\circ\rho(w)=\hat{\One}_{\Delta}(w)\One $.
\end{itemize}
This is a cointeraction in the sense studied in \cite{Fcointeraction, E-FFKP, Fchrom}.

\medskip

Note that with the product $m$, each $B^{[A_1, A_2]}$ is a monomial of such symbols where the $B$ part consists of a single edge.  In both \cite{Algebraic} Section C.4 and \cite{KlannMSc} the notation $x_{e, [A_1, A_2]}$ is used for these generators on the single edge $e$.  These generators are the fundamental cycles with cut information added in.

Defining these structures in terms of graph tree pairs is convenient for connections to fundamental cycles and to Cullen and Vogtmann's Outerspace \cite{CVouterspace, BeK, Kouterspace}, but results in an explosion of terms when doing concrete calculations, see \cite{KlannMSc}.  In Section~\ref{sec tree indep} we will give a formulation of this coaction which does not require fixing spanning trees.

\section{Analysis of the differential equations for the anomalous dimensions in $\phi^4$ theory}
Here we will consider the equations of the introduction in the special case of $\phi^4$ theory, \eqref{eq de1} and \eqref{eq de2}. We will study the existence of global solutions and Landau poles. Our analysis of the global solutions for the beta function will follow closely that of \cite{van_Baalen_2009}, whereas in order to study the Landau poles we present a different method based on Sturm Liouville. We also compare and classify the solutions and give some numerical examples of the different solutions. 
\subsection{Equations and set up}
Let us briefly show how to specialize the equations derived in \cite{YeatsThesis} into the set of ordinary differential equations which we will be analysing here. We first note that there are two different residues in massless $\phi^4$ theory, one for the vertex and one for the propagator. Thus we can identify the residue set with the set $\{+,-\}$ where $+$ refers to the vertex and $-$ to the propagator. Since the single vertex in the theory is four valent the combinatorial invariant charge as defined in \cite{YeatsThesis} becomes 
\[
    Q_{\phi^4} = \frac{X^{+}}{\left(X^{-}\right)^2}
\]
From which it is evident that $s_{+} = -1$ and $s_{-} = 2$. With these facts in mind it is easy to see that the equations for the two anomalous dimensions and the beta function in terms of the coupling constant $x$ become (rearranged from \eqref{eq de1} and \eqref{eq de2} for the first two)
\begin{equation}\label{eq:2}
    \frac{d\gamma^{+}}{dx} = \frac{\gamma^{+}(x)-\left(\gamma^{+}(x)\right)^2-P^{+}(x)}{x[\gamma^{+} (x) + 2\gamma^{-}(x)]}
\end{equation}
\begin{equation} \label{eq:3}
      \frac{d\gamma^{-}}{dx} = \frac{\gamma^{-}(x)+\left(\gamma^{-}(x)\right)^2-P^{-}(x)}{x[\gamma^{+} (x) + 2\gamma^{-}(x)]}
\end{equation}
and
\begin{equation} \label{eq:4}
    \beta(x) = x\left[\gamma^{+}(x)+2\gamma^{-}(x)\right].
\end{equation}

These are the main equations which we are going to analyse. In order to study this we will also need some rearrangements of the system of equations given in equations \ref{eq:2} - \ref{eq:4}. The first rearrangement comes by replacing the coupling constant $x$ with the variable $L$ defined through 
\begin{equation}\label{eq:16}
    \frac{dx}{dL} = \beta(x(L))
\end{equation}
A simple application of the chain rule shows that in terms of $L$ equations \ref{eq:2} - \ref{eq:4} become
\begin{equation}\label{eq:5}
    \frac{d\gamma^{+}}{dL} = \gamma^{+}(L)-\left(\gamma^{+}(L)\right)^2-P^{+}(L)
\end{equation}
\begin{equation} \label{eq:6}
      \frac{d\gamma^{-}}{dL} = \gamma^{-}(L)+\left(\gamma^{-}(L)\right)^2-P^{-}(L)
\end{equation}
Using these substitutions %, which we refer to as the parametric equations, 
we have eliminated the denominator of the differential equations. This is will be useful since it has simplified the equations and there are no longer problems in solving \ref{eq:5} and \ref{eq:6} when the beta function has a zero. One downside about all of the equations which we have presented so far is that they are all non-linear. There are multiple properties of linear differential equations which make them preferable for analysing see for example discussions in \cite{teschl2012ordinary}. Therefore we want to find a dependant variable which will make these equations linear. It is possible to do this for our case at the cost of increasing the order of each of the differential equations by one since equations \ref{eq:5} and \ref{eq:6} are Riccati equations. We therefore make the substitutions of the functions $u(L)$ and $v(L)$ defined by
\begin{equation} \label{eq:7}
    \gamma^{+} = \frac{1}{u(L)}\frac{du}{dL}
\end{equation}
\begin{equation} \label{eq:8}
    \gamma^{-} = -\frac{1}{v(L)}\frac{dv}{dL}
\end{equation}
In terms of the functions $u(L)$ and $v(L)$ the equations \ref{eq:5} and \ref{eq:6} become the following set of linear second order differential equations 
\begin{equation}\label{eq:9}
\frac{d^2 u}{dL^2}-\frac{du}{dL}+uP^{+}(L) = 0
\end{equation}
\begin{equation}\label{eq:10}
    \frac{d^2 v}{dL^2}-\frac{dv}{dL}-vP^{-}(L) = 0
\end{equation}
These two equations can be put into the self-adjoint form which will be helpful in some of our proofs and we thus remark now that the equations \ref{eq:9} and \ref{eq:10} can be written as 
\begin{equation}\label{eq:11}
    \frac{d}{dL}\left(e^{-L}\frac{du}{dL}\right) + e^{-L}P^{+}(L)u = 0
\end{equation}
\begin{equation}\label{eq:12}
    \frac{d}{dL}\left(e^{-L}\frac{dv}{dL}\right) - e^{-L}P^{-}(L)v = 0
\end{equation}
Each of these equations will help us in some way when finding solutions or proving things about the solutions. Before continuing on we will mention that there is a correspondence between Landau poles in $L$ of $\gamma^{+/-}$ and zeros of the solutions for $u$ and $v$. This correspondence will be stated precisely and proven in subsequent sections but for now it is worth noting that this substitution does much more than make the equations linear it will also help us find Landau poles.   
\subsection{Some closed form solutions}
In this section we will present some closed form solutions to these equations for certain special values of $P^{+/-}$.  These examples will give us some general intuition for the different solutions before we move on to making and proving general statements about the solutions. As far as we have found there are two possible choices for $P^{+/-}$ where we can find a closed form solution for $\gamma^{+/-}$. The two types of $P^{+/-}$ are $P^{+/-} = K^{+/-}$ where $K^{+/-}$ are constants and $P^{+/-} = \alpha_{+/-}L$.  We will consider these two cases separately beginning with the the case where $P^{+/-}$ are constants. 
\subsubsection{$P^{+/-}(x) = K^{+/-}$}

In this case equations \ref{eq:5} and \ref{eq:6} are directly integrable and the solutions are as follows. 

\[ %begin{equation}\label{eq:14}
\gamma^{-}(L) = \begin{cases}
a_{-}\left(\tanh\left[b^{-}-a^{-}(L-L_0)\right]\right)-\frac{1}{2} & K^{-}>-\frac{1}{4}\\
\frac{1}{c^{-}-(L-L_0)}-\frac{1}{2} & K^{-} = -\frac{1}{4} \\ 
\tilde{a}^{-}\left(\tan\left[\tilde{a}^{-}(L-L_0)+\tilde{b}^{-}\right]\right)-\frac{1}{2} & K^{-}<-\frac{1}{4}
\end{cases}
\] %end{equation}
\[ %begin{equation}
\gamma^{+}(L) = \begin{cases} %\label{eq:15}
a^{+}\left(\tan\left[b^{+}-a^{+}(L-L_0)\right]\right)+\frac{1}{2} & K^{+}>\frac{1}{4}\\
\frac{1}{c^{+}+(L-L_0)}+\frac{1}{2} & K^{+} = \frac{1}{4} \\ 
\tilde{a}^{+}\left(\tanh\left[\tilde{b}^{+}+\tilde{a}^{+}(L-L_0)\right]\right)+\frac{1}{2} & K^{+}<\frac{1}{4}
\end{cases}
\] %end{equation}
Where 
\begin{gather*}
a^{\pm} = \sqrt{K^{\pm}\mp\frac{1}{4}}, \qquad \tilde{a}^{\pm} = \sqrt{-\left(K^{\pm}\mp\frac{1}{4}\right)}, \qquad b^{+} = \textrm{arctan}\left(\frac{\gamma^{+}_0-\frac{1}{2}}{a^{+}}\right), \\ b^{-} = \textrm{arctanh}\left(\frac{\gamma^{-}_0+\frac{1}{2}}{a^{-}}\right), \qquad \tilde{b}^{+} = \textrm{arctanh}\left(\frac{\gamma^{+}_{0}-\frac{1}{2}}{\tilde{a}^{+}}\right), \\ \tilde{b}^{-} = \arctan\left(\frac{\gamma^{-}_{0}+\frac{1}{2}}{\tilde{a}^{-}}\right), \text{ and } c^{\pm} = \frac{1}{\gamma^{\pm}_0\mp\frac{1}{2}}.
\end{gather*}Of course there are nine possible explicit solutions for $\beta(L)$ depending on the values of $K^{+/-}$ as can be found using \eqref{eq:4}. We can also find $x(L)$ using \eqref{eq:16} and therefore we can find (in principle) $L(x)$ and therefore $\beta(x)$ although in practice the inversion might be difficult to do. It is worth pointing out here that already with these simple choices of $P^{+/-}$ there are solutions for the anomalous dimension of a function of $L$ with no Landau poles for any initial conditions, there are solutions which might or might not have Landau poles depending on initial conditions and there are solutions which have Landau poles regardless of the initial conditions. These examples will therefore be useful examples of the theorems which will be proven later. 

\medskip

Before moving onto the next example, consider another way of obtaining these solution, namely solve the equations \ref{eq:9} and \ref{eq:10} for $u$ and $v$ giving an equivalent solution to the one we found. The reason for solving these equations is that the solutions will be instrumental in some of the proofs later on and so it will be useful to have them outlined here. If the functions $P^{+/-}(L)$ are constant then the equations \ref{eq:9} and \ref{eq:10} reduce to the standard equations for a damped/driven oscillator and thus the solutions have the form 
\[ %begin{equation}\label{eq:17}
    u(L) = Ae^{k_{+}L} + Be^{k_{-}L}
\] %end{equation}
where $A,B\in \mathbb{R}$ and
\[ %begin{equation}\label{eq:18}
k_{+/-} = \frac{1\pm i\sqrt{4K^{+}-1}}{2}    
\] %end{equation}
as long as $K^{+}\neq 1/4$ since in this case $k_{+}=k_{-}$ and the two basic solutions fail to be linearly independent and thus fail to constitute a basis for solutions to the ODE. The solution in this case can be written as 
\[ %begin{equation}\label{eq:19}
    u(L) = e^{L/2}\left(A+BL\right)
\] %end{equation}
The solutions for $v(L)$ are analogously given by,
\[ %begin{equation}\label{eq:20}
    v(L) = Ae^{\kappa_{+}L} + Be^{\kappa_{-}L}
\] %end{equation}
where again $A,B\in \mathbb{R}$ and 
\[ %begin{equation}\label{eq:21}
    \kappa_{+/-} = \frac{1\mp \sqrt{4K^{-}+1}}{2}
\] %end{equation}
Again these solutions form a basis of the solution space unless $K^{-}=-1/4$ in which case the solution is given by, 
\[ %begin{equation}\label{eq:21}
    v(L) = e^{L/2}\left(A+BL\right)
\] %end{equation}
is the general solution. Notice how the values of $K^{\pm}$ which make the exponents imaginary and hence give growing oscillations and thus zeros of $u$ and $v$ and thus, as we will show, give Landau poles, correspond exactly to the solutions solved using the equation for $\gamma^{\pm}$ directly, as they should. Tables~\ref{tab:1} and \ref{tab:2} give a brief recap of the different types of solutions in both $\gamma^{+/-}$ and $u/v$ and weather the solutions in $L$ have Landau poles.
\begin{table}[]
    \centering
    \begin{tabular}{l|l|l|p{.4\linewidth}}
        $P^{+}(x(L))$ & $\gamma^{+}(L)$ & $u(L)$ & Landau Poles? \\
        \hline
        \hline
        $>\frac{1}{4}$  & $\propto \tan(L-L_0)$ & $\propto A\sin(kL)+B\cos(kL)$ & Exist for all initial conditions. \\
        \hline
        
        $ = \frac{1}{4}$ & $\propto \frac{1}{L-L_0}$ & $\propto e^{L/2}(A+BL)$ & Existence depends on initial conditions. \\ 
        \hline
        $<\frac{1}{4}$ & $\propto \tanh(L-L_0)$ & $\propto Ae^{k_{+}L} + Be^{k_{-}L}$ & No Landau poles regardless of initial conditions.
    \end{tabular}
    \caption{Summary of the different types of solutions for constant $P^{+}(L)$ depending on the values of $P^{+}(L)$.  }
    \label{tab:1}
\end{table}

\begin{table}[]
    \centering
    \begin{tabular}{l|l|l|p{.4\linewidth}}
        $P^{-}(x(L))$ & $\gamma^{-}(L)$ & $v(L)$ & Landau Poles? \\
        \hline
        \hline
        $<-\frac{1}{4}$  & $\propto \tan(L-L_0)$ & $\propto A\sin(kL)+B\cos(kL)$ & Exist for all initial conditions. \\
        \hline
        
        $ = \frac{1}{4}$ & $\propto \frac{1}{L-L_0}$ & $\propto e^{L/2}(A+BL)$ & Existence depends on initial conditions. \\ 
        \hline
        $>-\frac{1}{4}$ & $\propto \tanh(L-L_0)$ & $\propto Ae^{k_{+}L} + Be^{k_{-}L}$ & No Landau poles regardless of initial conditions.
    \end{tabular}
    \caption{Summary of the different types of solutions for constant $P^{-}(L)$ depending on the values of $P^{-}(L)$.  }
    \label{tab:2}
\end{table}
Now we will move on to consider solutions when $P^{+/-}(L)$ is of the form $P(L) = \alpha_{\pm}L$ 

\subsubsection{$P^{+/-}(L) = \alpha_{\pm}L$}

In this case it is easier to work with the linear second order differential equations of \ref{eq:9} and \ref{eq:10}.  The solutions for $u(L)$ and $v(L)$ are then linear combinations of Airy functions of the first and second kind of the form 
\[ %begin{equation}\label{eq:22}
    u(L) = Ae^{L/2}\textrm{Ai}\left(\frac{\frac{1}{4}-\alpha_{+}L}{\left(-\alpha_{+}\right)^{2/3}}\right) + Be^{L/2}\textrm{Bi}\left(\frac{\frac{1}{4}-\alpha_{+}L}{\left(-\alpha_{+}\right)^{2/3}}\right)
\] %end{equation}

\[ %begin{equation}\label{eq:23}
    v(L) =  Ae^{L/2}\textrm{Ai}\left(\frac{\frac{1}{4}+\alpha_{-}L}{\left(\alpha_{-}\right)^{2/3}}\right) + Be^{L/2}\textrm{Bi}\left(\frac{\frac{1}{4}+\alpha_{-}L}{\left(\alpha_{-}\right)^{2/3}}\right)
\] %end{equation}
Since as we have mentioned before that zeros of the functions $u$ and $v$ correspond to Landau poles it is worth looking at the zeros of both of these functions. Note that the Airy function of the first kind has an infinite number of zeros along the negative real axis and the Airy function of the second kind also has an infinite number of zeros along the negative real axis along with an infinite number of complex zeros. In principal any choice of the parameters gives solutions to the equation, but not all of these will be of interest to us because we wish to study only solutions which would physically make sense. There are good physical reasons to demand that the functions $\gamma^{+}$ and $\gamma^{-}$ are real functions and thus we demand this by demanding that $u$ and $v$ be complex functions whose derivatives have the same phase as the functions for all $L$. We also can demand that $L$ be real and positive and that $\alpha_{+}$ and $\alpha_{-}$ be real. With these assumptions we can restrict the type of zeros that we consider with the following argument. Since $\alpha_{\pm}$ are real we can write them as complex numbers in the following way $\alpha_{\pm} = |\alpha_{\pm}|e^{ik\pi}$ where $k$ is 1 or 2. Thus, $(-\alpha_{+})^{-2/3} = e^{i(1-2k/3)\pi}\left|\alpha_{+}\right|^{-2/3}$ and using the same manipulation we find $\alpha_{-}^{-2/3} = e^{-i2k\pi/3}\left|\alpha_{-}\right|^{-2/3}$. Since all other terms in the arguments of both Airy functions are real by assumption it follows that the phases of the arguments of the Airy functions are rational and therefore that these solutions have no complex zeros since none of the complex zeros of the Airy function have rational phase. Thus we can see that any solutions with $u$ and $v$ being complex have no Landau poles, whether or not they are physical solutions. 

\medskip

Based on this we can restrict ourselves to considering only real solutions for $u$ and $v$ which we do now. Since now we are only considering real solutions it follows $\alpha_{+}<0$ and $\alpha_{-}>0$. Since the zeros both Airy functions are negative there can only be a zero of $u$ or $v$ if there is some $L^{\ast}>L_0$ which has $A/B = - \textrm{Bi}\left(\frac{\frac{1}{4}+\alpha_{\pm}L^{\ast}}{\left(\alpha_{\pm}\right)^{2/3}}\right)/\textrm{Ai}\left(\frac{\frac{1}{4}+\alpha_{\pm}L^{\ast}}{\left(\alpha_{\pm}\right)^{2/3}}\right)$. It is simple to solve for the ratio $A/B$ in terms of the initial value for $\gamma^{\pm}$. The result is 
\[
\frac{A}{B} = -\frac{\gamma^{\pm}_{0}\textrm{Ai}\left(\frac{1/4+\alpha_{\pm}L_0}{\left(\alpha_{\pm}\right)^{2/3}}\right)+\textrm{Ai}'\left(\frac{1/4+\alpha_{\pm}L_0}{\left(\alpha_{\pm}\right)^{2/3}}\right)}{\gamma^{\pm}_{0}\textrm{Bi}\left(\frac{1/4+\alpha_{\pm}L_0}{\left(\alpha_{\pm}\right)^{2/3}}\right)+\textrm{Bi}'\left(\frac{1/4+\alpha_{\pm}L_0}{\left(\alpha_{\pm}\right)^{2/3}}\right)}.
\]
Thus the condition for the solution to have a Landau pole in this case becomes there exists some $L^{\ast}>L_0$ such that 
\begin{multline*}
    \textrm{Ai}\left(\frac{\frac{1}{4}+\alpha_{\pm}L^{\ast}}{\left(\alpha_{\pm}\right)^{2/3}}\right)\left[\gamma^{\pm}_{0}\textrm{Ai}\left(\frac{1/4+\alpha_{\pm}L_0}{\left(\alpha_{\pm}\right)^{2/3}}\right)+\textrm{Ai}'\left(\frac{1/4+\alpha_{\pm}L_0}{\left(\alpha_{\pm}\right)^{2/3}}\right)\right]\\ =  \textrm{Bi}\left(\frac{\frac{1}{4}+\alpha_{\pm}L^{\ast}}{\left(\alpha_{\pm}\right)^{2/3}}\right)\left[\gamma^{\pm}_{0}\textrm{Bi}\left(\frac{1/4+\alpha_{\pm}L_0}{\left(\alpha_{\pm}\right)^{2/3}}\right)+\textrm{Bi}'\left(\frac{1/4+\alpha_{\pm}L_0}{\left(\alpha_{\pm}\right)^{2/3}}\right)\right]
\end{multline*}
This concludes our two examples of closed form solutions we are going to study. In the next section we generalize the approaches used to study these solutions to make more general statements about the existence of Landau poles when the solutions are not necessarily expressible in closed form. 

\subsection{Existence of Landau poles}
One of our main goals in studying these equations is to determine when the solutions diverge in finite $L$ or $x$. Solutions where this happens for finite $L$ are called Landau poles. For solutions which diverge in finite $x$ we will simply call these  divergences, though some references call these Landau poles as well. These types of solutions occur in many quantum field theories, most famously the singularity in the QED beta function and are generally considered failures of the theory since they predict that the coupling constant will diverge at some finite scale which is nonphysical. Here we give conditions under which the theory we are studying will have or not have these Landau poles.

\medskip

We will consider the two cases of divergence in $x$ and $L$ separately, beginning with studying these as functions of $L$ since these are the Landau poles we're interested in. We have already mentioned and from this point on will frequently use that zeros of $u$ and $v$ correspond to Landau poles in $\gamma^{+}$ and $\gamma^{-}$. We now give a proof of this fact. 
\begin{lemma}\label{lem:1}
Let $u$ and $\gamma^{+}$ and $v$ and $\gamma^{-}$ be related via equations \ref{eq:7} and \ref{eq:8} and suppose that $u$ and $v$ are analytic. Then there is some $L^{*}<\infty$ such that $\textrm{lim}_{L\rightarrow L^{\ast}} \gamma^{\pm}(L) = \infty$ if and only if $L^{\ast}$ is a zero of $u$ or $v$ respectively. 
\end{lemma}
\begin{proof}
First suppose that $L^{\ast}$ is a zero of $u$ or $v$ since $u$ and $v$ are analytic, at $L^{\ast}$ there is an open interval $(L^{\ast}-\delta,L^{\ast}+\delta)$ where $u$ can be represented by a convergent Taylor series of the form 
\begin{equation*}
    u(L) = \sum_{n=1}^{\infty} \frac{u^{(n)}(L^{\ast})}{n!}(L-L^{\ast})^n
\end{equation*}
and thus 
\begin{equation*}
    u'(L) = \sum_{n=0}^{\infty} \frac{u^{n+1}(L^{\ast})}{n!}(L-L^{\ast})^n.
\end{equation*}
The term by term differentiation is justified by the fact that $u$ is analytic. Thus in this same open neighbourhood, 
\begin{equation*}
    \gamma^{+} = \frac{1}{L-L^{\ast}}\frac{\sum_{n=0}^{\infty} \frac{u^{n+1}(L^{\ast})}{n!}(L-L^{\ast})^n}{\sum_{n=0}^{\infty} \frac{u^{(n+1)}(L^{\ast})}{(n+1)!}(L-L^{\ast})^n} = \frac{f(L)}{L-L^{\ast}}.
\end{equation*}
Where $f(L)$ is a function with no zeros at $L^{\ast}$ and therefore we find that $\gamma^{+}$ diverges as $L\rightarrow L^{\ast}$.  The same proof applies to $\gamma^{-}$ by replacing $u$ with $v$. Conversely since $u$ and $v$ are analytic, by the Taylor series expansions we see that the derivatives of $u$ and $v$ can't diverge and thus either $u$ or $v$ must tend to zero.
\end{proof}
By Lemma~\ref{lem:1} in order to find when there are Landau poles in $\gamma^{\pm}$ as a function of $L$ it suffices to find the zeros of the solutions for $u$ and $v$. For a function satisfying a second order linear differential equation like $u$ and $v$ do it is possible to use Sturm-Liouville theory to determine information on the zeros of these functions which is what we will do. The next theorem demonstrates this. 
\begin{theorem}\label{thm:1}
Suppose that $P^{+}(L)>\frac{1}{4}$ on some interval $I = [L_0,L^{\ast}]$ and let $\rho^{+} = \textrm{min}_{L\in I} P^{+}(L)$ then if $J=[L_0,L_{0}+\frac{2\pi}{k}]\subseteq I$ where $k = \frac{\sqrt{4\rho^{+}-1}}{2}$ then $\gamma^{+}$ has a Landau pole in $I$. 
\end{theorem}
\begin{proof}
Let $\tilde{u}(L)$ be the solution to \ref{eq:11} with $P^{+}(L) = \rho^{+}$ then from the previously presented solutions we know 
\begin{equation*}
    \tilde{u}(L) = e^{\frac{L-L_0}{2}}\left[\tilde{u}(L_0)\cos\left(k(L-L_0)\right) + \frac{\tilde{u}'(L_0)}{k}\sin\left(k(L-L_0)\right) \right]
\end{equation*}
It is elementary to observe that $\tilde{u}(L)$ has two zeros both contained in $J$ and therefore has two or more zeros in $I$, let $u(L)$ be the solution to \eqref{eq:11} with $P^{+}(L)$ be as in the theorem then since $P^{+}(L)\geq\rho^{+}$ by the Sturm comparison theorem it follows $u(L)$ has at least one zero in $I$ and by the lemma $\gamma^{+}$ must also have a Landau pole in $I$.
\end{proof}

Almost the exact same proof will work for $v$ except the condition for the solutions is slightly different. Below we state the theorem and explain the slight difference in the proof. 

\begin{theorem}\label{thm:2}
Suppose that $P^{-}(L)<-\frac{1}{4}$ on some interval $I = [L_0,L^{\ast}]$ and let $\rho^{-} = \textrm{max}_{L\in I} P^{-}(L)$ then if $J=[L_0,L_{0}+\frac{2\pi}{k}]\subseteq I$ where $k = \frac{\sqrt{1-4\rho^{-}}}{2}$ then $\gamma^{+}$ has a Landau pole in $I$.
\end{theorem}
\begin{proof}
The proof is exactly the same as Theorem \ref{thm:1} except one applies the Sturm comparison theorem to $-1\times$ \eqref{eq:11} where one $P^{-}(L)$ is $\rho^{-}$ and the other is our given $P^{-}(L)$ 
\end{proof}

It is also worth pointing out that we can weaken the hypotheses of the previous two theorems slightly and still guarantee the existence of Landau poles which we describe now.
\begin{theorem}\label{thm:3}
Suppose that there is some $L^{\ast}$ such that $\rho^{+}_{\infty} = \textrm{inf}_{L>L^{\ast}}P^{+}(L)>\frac{1}{4}$ then $\gamma^{+}$ has a Landau pole for some $\bar{L}>L^{\ast}$
\end{theorem}

\begin{proof}
Let $k=\frac{\sqrt{4\rho^{+}_{\infty}-1}}{2}$ we again apply the Sturm comparison theorem on the two solutions to \eqref{eq:11} with the given $P^{+}(L)$ and the solution to \eqref{eq:11} where $P^{+}(L) = \rho^{+}$ on the interval $[L^{\ast},L^{\ast} + 2\pi/k]$ since the solution with $P^{+}(L) =\rho^{+}$ has two zeros in this interval the solution with our given $P^{+}(L)$ must have at least one in the same interval. 
\end{proof}

Once again the same theorem will apply to $v$ as well with a slightly different condition which is given below. 
\begin{theorem}\label{thm:4}
Suppose that there is some $L^{\ast}$ such that $\rho^{+}_{\infty} = \textrm{inf}_{L>L^{\ast}}P^{+}(L)<-\frac{1}{4}$ then $\gamma^{-}$ has a Landau pole for some $\bar{L}>L^{\ast}$
\end{theorem}
\begin{proof}
The proof is the same as Theorem \ref{thm:3} except that we apply the Sturm comparison theorem with -\eqref{eq:11} in the same way as the previous theorem. 
\end{proof}
We now discuss the existence of divergences in the variable $x$, but in terms of the opposite question. Specifically when are there global solutions to the beta function as a function of $x$. This is interesting physically because it means the theory it corresponds to the theory being well behaved for all values of the coupling $x$.    
\subsection{Global Solutions of the Beta Function}
It will be useful to relate this to the beta function alone rather than both of the anomalous dimensions separately. In order to do this we will give an equation for the the beta function by combining equations \ref{eq:3} and \ref{eq:4} into a single differential equation. The result of this is the following equation 
\begin{equation}\label{eq:13}
    \frac{d\beta}{dx} = 1+\frac{2\beta(x)}{x}-x\frac{2(\gamma^{+}+\gamma^{-})^2+Q(x)}{\beta(x)}
\end{equation}
Where $Q(x) := P^{+}(x)+2P^{-}(x)$ In order to simplify the rest of the analysis we will need another lemma which will help us to focus on physical solutions. For a physical solution we want the beta function to be single valued and it is easy to find a condition where this won't happen and we therefore find a non physical solution for the beta function. 

\begin{lemma}
Suppose $x^{\ast}>0$ is a point such that $\beta(x^{\ast}) = 0$ and that $Q(x)>0$  then $x^{\ast}$ is a maximum of $x(L)$.
\end{lemma}
\begin{proof}
Since $x^{\ast}$ is a zero of $\beta(x)$ by \eqref{eq:5} it follows that $x^{\ast}$ is a maximum or minimum by considering \ref{eq:13} and using the fact that $\beta(x^{\ast}) =0 $ it follows that 
\begin{equation*}
    \frac{d^2 x}{dL^2} = \beta\frac{d \beta}{dx} = -2x^{\ast}(\gamma^{+}+\gamma^{-})^2-x^{\ast}Q(x^{\ast})<0\end{equation*}
\end{proof}
This lemma means that if we are only interested in physical solutions (as we are) it is enough to consider solutions for the beta function which are strictly positive or negative since otherwise we will have solutions where there is a maximum possible $x$ as a function of $L$ meaning that our beta function must be multi valued. With this out of the way, we can discuss the theorem relating to the existence of global solutions for the beta function in $x$.
\begin{theorem}\label{thm:thm5}

Let $Q(x)$ be a $C^2$ positive function with $\gamma^{-}>0$ and $2\gamma^{+}+3\gamma^{-}$ having the same sign, then global solutions to (3.12) exist if there is some $x_0$ such that  
\begin{equation}\label{eq:30}
\int_{x_0}^{\infty}\: zQ(z) \: dz <\infty
\end{equation}
\end{theorem}
This theorem gives an integrability condition for global solutions to exist. The proof will follow \cite{van_Baalen_2009} very closely.
\begin{proof}
First, let $x_0$ be as in the theorem statement, $\gamma^{\pm}_0 = \gamma^{\pm}\left(x_0\right)$ and $\epsilon>0$ choose
\begin{equation}\label{init_conds_eq}
\gamma^{+}_0 + 2\gamma^{-}_0  = \frac{1}{x_0}\left(2\int_{x_0}^{\infty}\: zQ(z)\; dz +\epsilon^2\right)^{\frac{1}{2}}
\end{equation}
 
Note that for global solutions to not exist we must have that for some $x^{\ast}<\infty$ we have either $\gamma^{+}\left(x^{\ast}\right) + 2\gamma^{-}\left(x^{\ast}\right) \rightarrow \infty$ or  $\gamma^{+}\left(x^{\ast}\right) + 2\gamma^{-}\left(x^{\ast}\right) = 0$. The solution can't have a pole since equation \ref{eq:30} bounds the solution,thus assume for a contradiction that $\gamma^{+}\left(x^{\ast}\right) + 2\gamma^{-}\left(x^{\ast}\right) = 0$ so the third term is positive too, now consider,
\begin{equation*}
\frac{1}{2}\frac{d}{dx}\left(\gamma^{+}+2\gamma^{-}\right)^2 = \frac{\gamma^{+}+2\gamma^{-}}{x}+ \frac{\left(\gamma^{+}+2\gamma^{-}\right)^2}{x} - \frac{2\left(\gamma^{+}+\gamma^{-}\right)^2}{x} - \frac{Q(x)}{x} 
\end{equation*}
Simplifying we have 
\begin{equation*}
\frac{1}{2}\frac{d}{dx}\left(\gamma^{+}+2\gamma^{-}\right)^2 = \frac{\gamma^{+}+2\gamma^{-}}{x} - \frac{\left(\gamma^{+}+2\gamma^{-}\right)^2}{x} + \frac{2\gamma^{-}\left(2\gamma^{+}+3\gamma^{-}\right)}{x}-\frac{Q\left(x\right)}{x}
\end{equation*}
Or, 
    \begin{equation*}
\frac{1}{2}\frac{d}{dx}\left(\gamma^{+}+2\gamma^{-}\right)^2 \geq - \frac{\left(\gamma^{+}+2\gamma^{-}\right)^2}{x} -  \frac{Q\left(x\right)}{x}
\end{equation*}
Where the inequality follows since both terms removed are positive rearranging gives,
\begin{equation*} 
x^2\frac{d}{dx}\left(\gamma^{+}+2\gamma^{-}\right)^2 + 2x\left(\gamma^{+}+2\gamma^{-}\right)^2 \geq -2x Q (x)
\end{equation*} 
Or, 
\begin{equation} \label{integral_eq}
\frac{d}{dx}\left(x^{2}\left[\gamma^{+}+2\gamma^{-}\right]^2\right) \geq -2 x Q(x)
\end{equation}
Integrating equation \ref{integral_eq} on $\left[x_0,x^{\ast}\right]$ and using \ref{init_conds_eq} gives
\begin{equation}
\gamma^{+}+2\gamma^{-}\geq \frac{1}{x^2}\left(x_{0}^2 \left[\gamma^{+}_0 + 2\gamma^{-}_0\right]-\int_{x_0}^{x^{\ast}}\: zQ(z) \; dz \right)  > \left(\frac{\epsilon}{x_0}\right)^2 > 0 
\end{equation}
Contradicting our assumption that $\gamma^{+}\left(x^{\ast}\right) + 2\gamma^{-}\left(x^{\ast}\right) = 0$. This shows that given our assumptions \ref{eq:30} is enough to give global solutions.
\end{proof}

This concludes our study of the existence of Landau poles. In the final part of this section we will give a summary of these results and give some numerical examples of the solutions to the equations. 

\subsection{Numerical results and examples}
To begin this section we give a brief summary of the relationship between the variables $x$ and $L$. Notice that by \ref{eq:16} we have that 
\begin{equation*}
    x = x_0 + \int_{L_0}^{L} \beta(z) \: dz
\end{equation*}
Thus if there is a Landau pole in $L$ we see that $x\rightarrow \infty$ and thus there is no divergence in $x$ since it it defined for all positive values. Conversely if there is a divergence in $x$ then we can see that the rearranging the same equation gives 
\begin{equation*}
    L = L_0 + \int_{x_0}^{x}\frac{1}{\beta(x)} \: dx
\end{equation*}
Thus we can see that if $x^{\ast}$ is a divergence in $x$ then assuming that there are no zeros of $\beta(x)$ in $[x_0,x^{\ast}]$ this will be a Landau pole of $\beta(L)$ as well since the integral defining $L$ will converge to a finite value as $\beta(x)\rightarrow \infty$ which can be seen easily using the mean value theorem for definite integrals.  

\medskip

Using these two relationships between Landau poles and divergences we are going to tabulate the results of our theorems in this section as way of recapping them all together 

\medskip

\begin{tabular}{c|c|c}
    Conditions of Thm X are Satisfied & $\beta(x)$ & $\beta(L)$ \\ 
    \hline
     \ref{thm:1}& Global Solutions   & Landau Pole \\
     
     \ref{thm:thm5} & Global Solutions & No Landau Pole
\end{tabular}

Having these general theorems proven we now move on to some numerical examples of the types of solutions which we can find. Even with the simple case where $P^{\pm}(L)$ is a constant we can find all three possible types of solutions which can possibly happen. As we have seen there are three possible types of solutions, there are solutions with no Landau poles for any initial conditions, solutions where there are Landau poles for any initial conditions and solutions where there may or may not be Landau poles depending on initial conditions. Figure~\ref{fig:figure1} shows three different solutions to the equations which exhibit each of these behaviors. This could be deduced from the mathematical form of the solutions alone but this also provided a useful way of verifying our code which generated the numerical solutions in the case where we do not have a closed form solution. Our first numerical example is simply to serve as a test that our code to generate numerical solutions work well as well to illustrate that even in the simple case where $P(L)$ is a constant there are solutions with a single Landau pole, an infinite number of Landau poles and solutions with no Landau poles. 
\begin{figure}
    \centering
    \includegraphics[width=\textwidth,height=20cm,keepaspectratio]{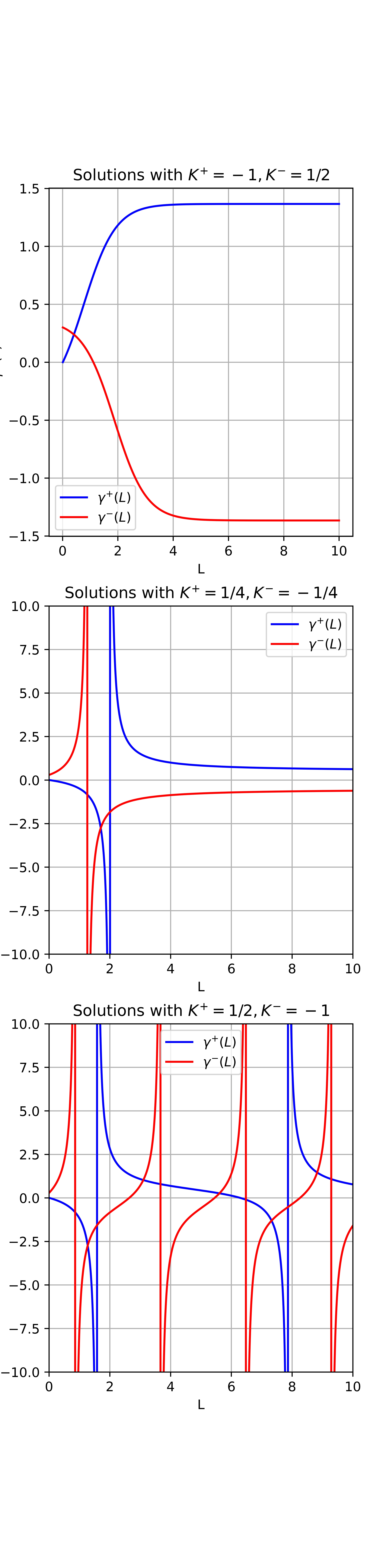}
    \caption{Three different solutions plotted for three different constant values for  each showing the different possibilities for existence of Landau poles.}
    \label{fig:figure1}
\end{figure}

Next we give some examples of the solutions for the functions for $u(L)$ again showing the different behavior of Landau poles. Figure~\ref{fig:figure2} gives an example without a Landau pole which was generated from a constant $P^{+}(L) = -1$ which as we already know will have no Landau pole. Even in cases where we don't have a closed form solution for the anomalous dimension these solutions for $u$ could also tell us there are no Landau poles since the solution doesn't cross zero.  
\begin{figure}
    \centering
    \includegraphics[width=\linewidth]{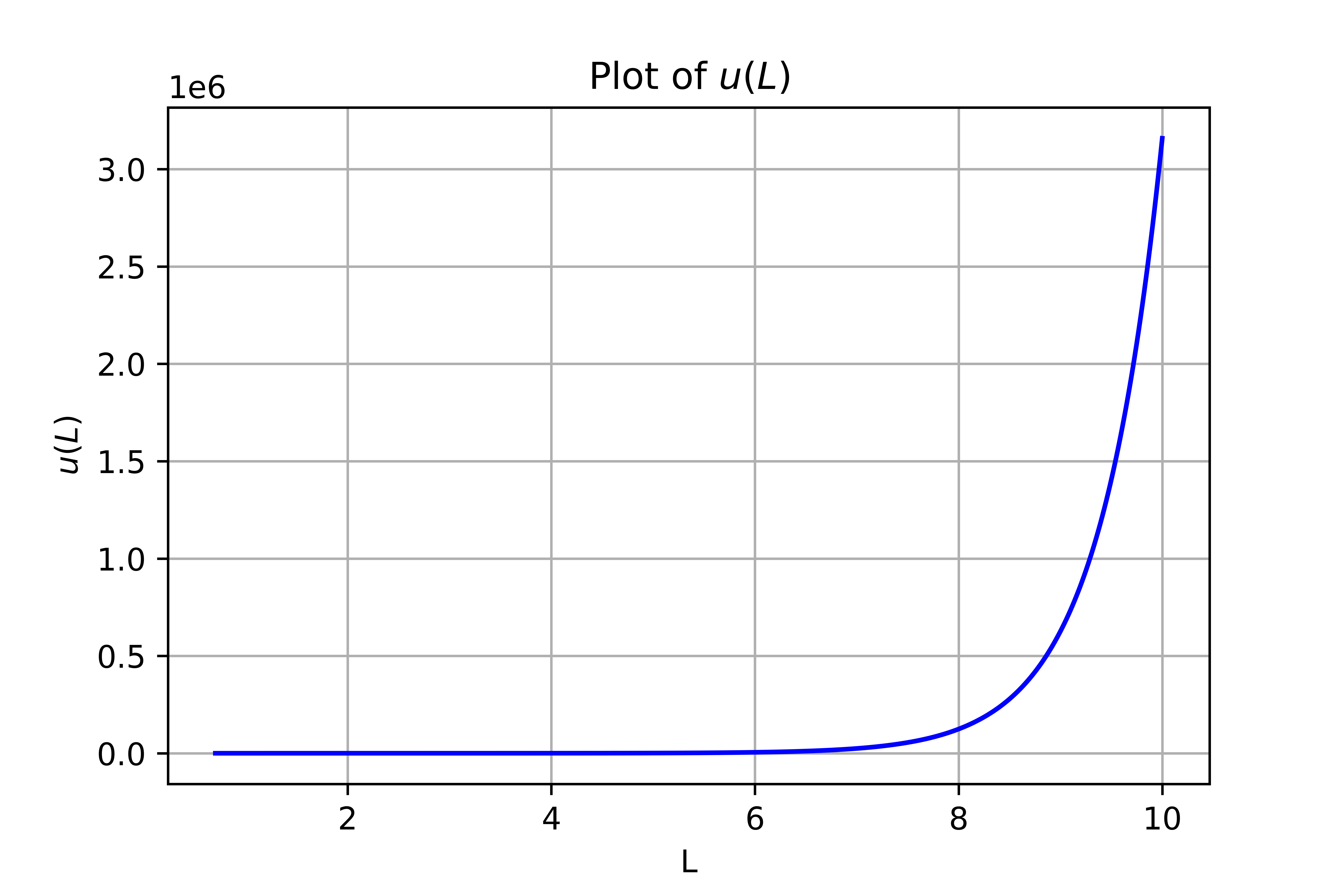}
    \caption{A solution for $u(L)$ which doesn't have a Landau pole which can be seen by the fact that the solution for $u(L)$ never crosses zero. }
    \label{fig:figure2}
\end{figure}

Figure~\ref{fig:figure3} gives an example with an infinite number of Landau poles in a solution for $\gamma^{\pm}$ which we can see from the solution for $u(L)$ since the solution for $u(L)$ is oscillatory and crosses zero multiple times (in fact infinite times) each of these is a Landau pole and in fact this is another solution with constant $P^{+}(L)=1$ and we already know this solutions has an infinite number of Landau poles and these correspond to the zeros of the oscillatory solution for $u(L)$.
\begin{figure}
    \centering
    \includegraphics[width=\linewidth]{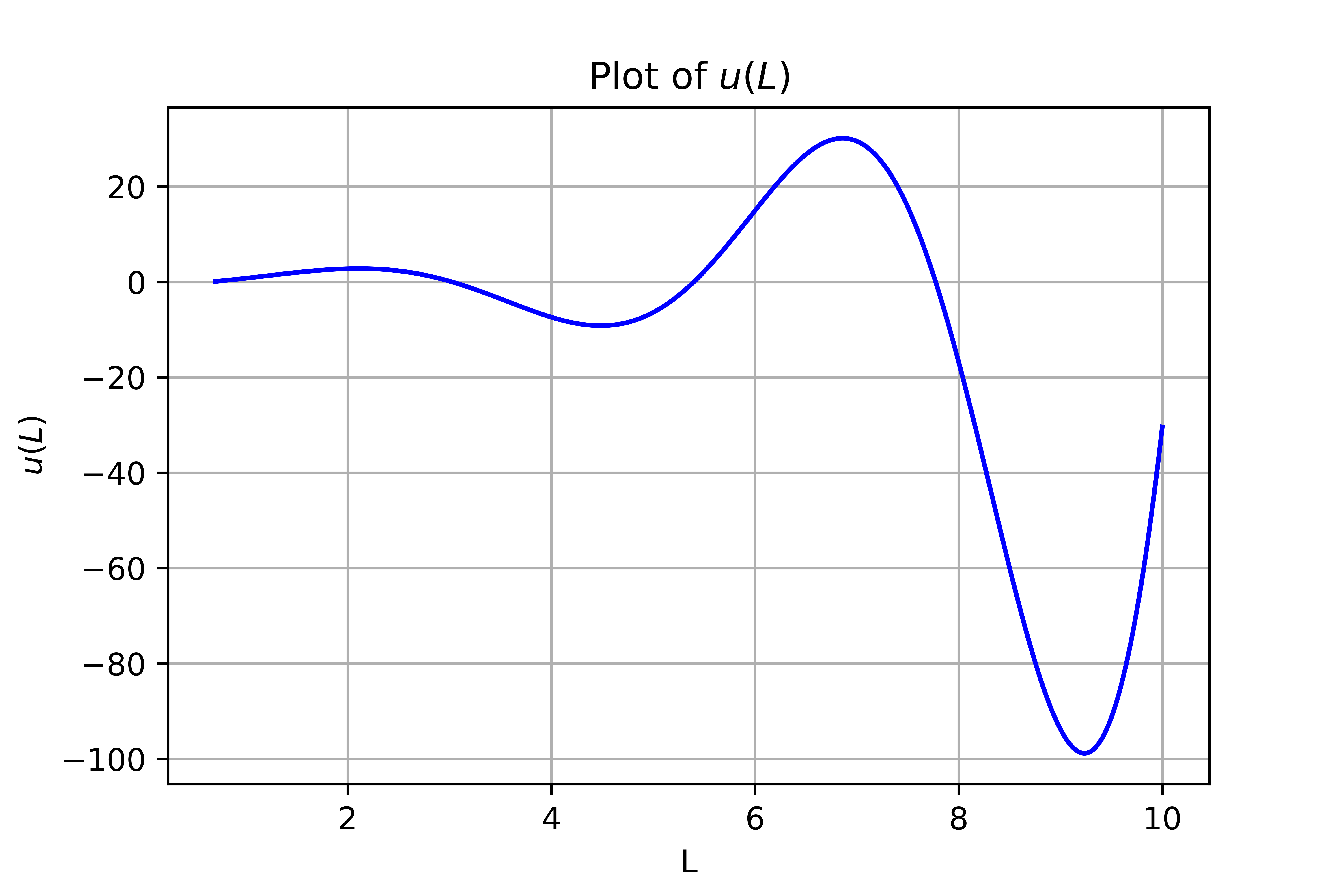}
    \caption{A solution for $u(L)$ where the corresponding solution for the anomalous dimension has an infinite number of Landau pole as can be seen by the oscillatory solution}
    \label{fig:figure3}
\end{figure}

In addition to this we also investigated the different behaviours for different initial conditions and numerically with a fixed $P^{+}(x) = 0.0833x^2 + 0.1874x^3$ and $P^{-}(x) = 1.5x + 1.8336x^2 - 3.728x^3$ and here we see an interesting result that there is an almost linear border between solutions which have Landau poles and those which do not. See Figure~\ref{fig colour plot}. We also note that this behavior persists with different initial conditions for $x$, see Figure~\ref{fig colour plot}. These plots show that indeed the regions of $\left(\gamma^{+}_0,\gamma^{-}_0\right)$ which produce divergences in $x$ appear to be smoothly separated from those which don't. As we expect solutions which have initial opposite signs will eventually have divergences (this is because the divergence happens when the denominator of \ref{eq:2} or \ref{eq:3} vanishes). It is left to a future work to give a complete description of these boundaries between solutions with divergences and without divergences in this and other models.

\begin{figure}
	\centering
	
	\includegraphics[scale=0.6]{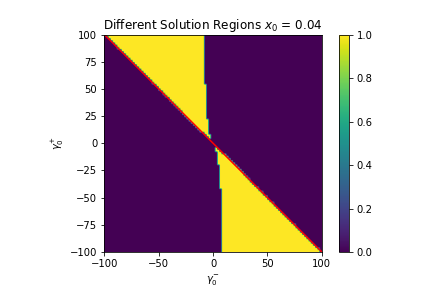}
	\includegraphics[scale=0.6]{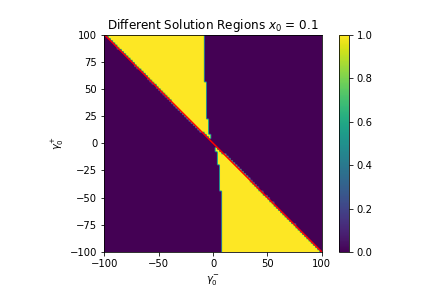}
	\includegraphics[scale=0.6]{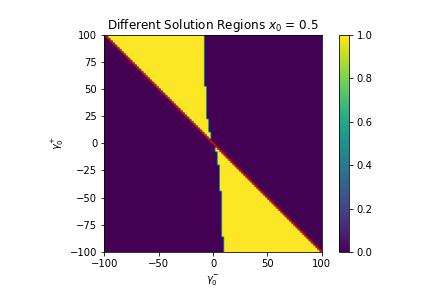}

	\caption{An example of the regions generated by the different initial conditions. The yellow region has solutions with one non-trivial zero of the beta function and the purple region has no non-trivial zeros.
	The green line is the border between solutions without divergences in $x$ and solutions which don't diverge in finite $x$. The red  line is the line $\gamma^{+}_0 = -\gamma^{-}_0$.}
	\label{fig colour plot}
\end{figure}

\section{The cut/core cointeraction}\label{sec coact}
\subsection{Tree independent formulation of the coaction}\label{sec tree indep}
Let us now return to the observation that the description of the cointeration of the cut bialgebra and the core Hopf algebra in \cite{Algebraic} was defined in terms of graph spanning tree pairs, which has some benefits, particularly for relating to other mathematical constructions, but also a disadvantage in the explosion of terms.

For the coproduct $\Delta$ of Section~\ref{subsec orig coact} the formulation without a spanning tree is simply the core coproduct, see Definition \ref{def core} and \eqref{eq modified core coprod}.  This could be paired with either the usual disjoint union product of the core Hopf algebra or an internal product similar to $m$ of Section~\ref{subsec orig coact}.  In the latter case, we would have a fixed mother graph $X$ and all other graphs $G$ under consideration would be subgraphs of $X$.  These subgraphs would not be considered up to isomorphism, rather they would maintain the information of how they are subgraphs of $X$.  Then the product would be union in $X$.

More interesting is how to define $\rho$ without a spanning tree.  To state and prove the result we will need some notation.  

In Section~\ref{subsec orig coact}, we observed that the vector space spanned by graphs is embedded in the vector space spanned by graph spanning tree pairs by summing over all spanning trees.  Let us now give this embedding a name.
\begin{definition}
For a connected graph $G$ define
\[
s(G) = \sum_{\substack{T \text{ spanning} \\ \text{tree of } G}} (G,T)
\]
\end{definition}

Also, recall from Section~\ref{subsec orig coact}, that if we had a graph tree pair $(G,T)$ and a subset $B$ of edges not in $T$ then we defined $\gamma_B$ to be the subgraph defined by taking the union of the fundamental cycles defined by the edges of $B$.  Recall further that in Section~\ref{subsec orig coact} for each $(G,T)$ we had a space of formal symbols $B^{[A_1, A_2]}$.  We had discussed the intended interpretation of these symbols, but now we need a name for the map instantiating this interpretation.  To this end define the following two things.  First given $(G,T)$ and $A\subseteq E(T)$ let
$C_A$
be the edges of the cut induced by $A$; that is, $C_A$ is those edges whose ends are in different components of $T-A$.  Next, we make the following definition of the interpretation of the symbol.
\begin{definition}
Given $(G,T)$, $B\subseteq E(G)-E(T)$, and $A_1, A_2, \subseteq E(T)$, define
\[
g(B^{[A_1, A_2]}) = \big((\gamma_B - C_{A_1})/(T-A_2), T\cap (A_2-A_1)\big).
\]
\end{definition}
To say this another way, $g(B^{[A_1, A_2]})$ is the pair of 
\begin{enumerate}
    \item the graph defined by $B$ with the cut defined by $A_1$ done and the edges in $E(T)-A_2$ contracted, along with 
    \item the spanning forest given by those edges of $T$ that are still present in the graph.
\end{enumerate} 
These remaining edges of $T$ give a spanning tree of each component of the graph.
Note that $g(B^{[A_1, A_2]})$ depends on the $(G,T)$ relative to which $B^{[A_1, A_2]}$ is defined even though $G$ and $T$ do not appear in the notation.

We will also use $\pi_1$ to denote the projection onto the first component of a pair.  In particular $\pi_1(G,T) = G$.  As a final prelimiary definition, it is time to formally specify what we mean by a cut of a graph
\begin{definition}
A \emph{cut} of a graph $G$ is a set of edges $C$ so that every edge of $C$ has its ends in different components of $G-C$.  If the connected components of $G-C$ are $G_1, G_2, \ldots, G_k$ then we say we have a cut of $G$ into $G_1, G_2, \ldots, G_k$.
\end{definition}

With these definitions at hand we can now define a version of $\rho$, which we call $\widetilde{\rho}$, that acts on $G$ not $(G,T)$.  
\begin{definition}
For a connected graph $G$,
\[
\widetilde{\rho}(G) = \sum_{\substack{C \text{ cut of } G \text{ into}\\ G_1, G_2, \ldots, G_k}}\left(\sum_{\substack{T_i \text{ spanning}\\\text{tree of } G_i}} G/(T_1\cup T_2\cup \cdots \cup T_k)\right) \otimes G_1\cup G_2\cup \cdots \cup G_k
\]
\end{definition}

The claim now is that this $\widetilde{\rho}$ should be seen as the tree independent formulation of $\rho$.  The formal statement is as follows.
\begin{theorem}
\[
(s\otimes \text{id})\circ \widetilde{\rho}(G) = (\text{id}\otimes \pi_1)\circ (g\otimes g)\circ\rho \circ g^{-1} s(G)
\]
where $g^{-1}(G,T)$ is interpreted relative to $(G,T)$; specifically, 
\[
g^{-1}(G,T) = (E(G)-E(T))^{[\emptyset, E(T)]}
\]
as a symbol relative to $(G,T)$.
\end{theorem}

\begin{proof}
The proof is by direct computation.  Let $E_L = E(G)-E(T)$.  Then
{\allowdisplaybreaks
\begin{align*}
    & (\text{id}\otimes \pi_1)\circ (g\otimes g)\circ\rho \circ g^{-1} s(G) \\
    & = (\text{id}\otimes \pi_1)\circ (g\otimes g)\circ\rho\left(\sum_{\substack{T \text{ spanning}\\\text{tree of } G}}E_L^{[\emptyset, E(T)]}\right) \\
    & = (\text{id}\otimes \pi_1)\sum_{\substack{T \text{ spanning}\\\text{tree of } G}} (g\otimes g) \left(\sum_{\emptyset \subseteq A \subseteq E(T)} E_L^{[\emptyset, A]} \otimes E_L^{[A, E(T)]} \right)\\
    & =(\text{id}\otimes \pi_1) \sum_{\substack{T \text{ spanning}\\\text{tree of } G}} \left(\sum_{\emptyset \subseteq A \subseteq E(T)} (G/(T-A), A) \otimes (G-C_A, T-A) \right)\\
    & =(\text{id}\otimes \pi_1)  \sum_{\text{cut } C \text{ of } G} \sum_{\substack{T \text{ spanning tree}\\\text{of } G \text{ such that}\\\text{some subset}\\\text{of } E(T) \text{ induces } C}}(G/(T-(T\cap C)), T\cap C) \otimes (G-C, T-C) \\
    & = \sum_{\text{cut } C \text{ of } G} \Bigg(\sum_{\substack{T \text{ spanning tree}\\\text{of } G \text{ such that}\\\text{some subset}\\\text{of } E(T) \text{ induces } C}}(G/(T-(T\cap C)), T\cap C)\Bigg) \otimes (G-C) \\
    & =\sum_{\substack{\text{cut of } G \text{ into}\\ G_1, G_2, \ldots, G_k}} \Bigg(\sum_{\substack{T_i \text{ spanning}\\\text{tree of } G_i \\ A \text{ spanning tree} \\\text{of } G/(T_1\cup \cdots \cup T_k)}}(G/(T_1\cup\cdots \cup T_k)), A)\Bigg) \otimes (G_1\cup\cdots \cup G_k)\\
    & =\sum_{\substack{\text{cut of } G \text{ into}\\ G_1, G_2, \ldots, G_k}}\left( \sum_{\substack{T_i \text{ spanning}\\\text{tree of } G_i }}\left( \sum_{\substack{ A \text{ spanning tree} \\\text{of } G/(T_1\cup \cdots \cup T_k)}}(G/(T_1\cup\cdots \cup T_k)), A) \right)\right)\otimes (G_1\cup\cdots \cup G_k)\\
    & = \sum_{\substack{\text{cut of } G \text{ into} \\ G_1, G_2, \ldots, G_k}} \left(\sum_{\substack{T_i \text{ spanning}\\\text{tree of } G_i }} s(G/(T_1\cup \cdots \cup T_k))\right)\otimes (G_1\cup\cdots \cup G_k)\\
    & = (s\otimes \text{id})\widetilde{\rho}(G)
\end{align*}
}
The key observation underlying the calculation is that the information of a spanning tree and a cut induced by a subset of edges of the spanning tree is the same as the information of a spanning tree of each component after the cut along with a spanning tree of the graph with each of those components contracted.
\end{proof}

If we wish to have certain edges of the graph not be able to be tadpoles, then we simply forbid cuts and trees which would cause such tadpoles.  The most physically relevant case is when we disallow fermion tadpoles, in which case the cuts in the definition of $\widetilde{\rho}$ should be required to cut all internal fermion loops and the $T_i$ in the definition of $\rho$ should be required to use all remaining fermion edges.  To see that this is the correct condition, first suppose an internal fermion loop was uncut.  Then this fermion loop would be entirely in one of the $G_i$, but then any spanning tree $T_i$ of $G_i$ does not contain at least one edge of the fermion loop, and since contracting a tree of $G_i$ leaves the rest of the edges of $G_i$ as tadpoles, then we get a fermion tadpole.  On the other hand if all fermion loops are cut, then in each $G_i$ the remaining fermion edges contain no cycles so at least one spanning tree exists that uses all the fermion edges, and there will be fermion tadpoles precisely if some fermion edge is not in the tree.

Note that cutting all fermion loops doesn't forbid cutting the same loop additional times, if some non-fermion edges are appropriately placed.  For example consider cutting the four fermion edges slashed in red in Figure~\ref{fig:QEDcut}.  This is a valid cut for the action of $\widetilde{\rho}$ resulting in two connected components with one external edge each, but cuts the same fermion loop multiple times.

\begin{figure}
    \centering
    \includegraphics{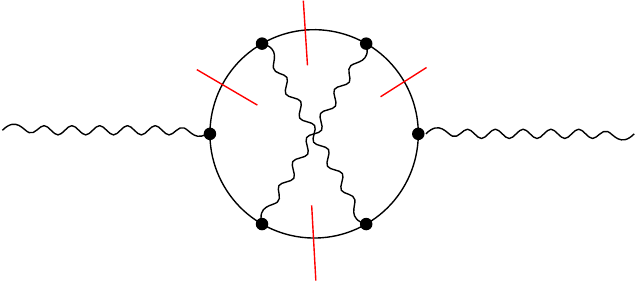}
    \caption{A cut that cuts the same fermion loop multiple times.}
    \label{fig:QEDcut}
\end{figure}

\subsection{Combinatorial considerations towards the interplay between $\rho$ and IR limits with soft photons}

Let us now consider some of the combinatorial considerations that would underlie generalizing the calculations of Chapter 4 of \cite{KlannMSc}.  To understand that chapter, we first need the \emph{Galois coaction} which is constructed by first applying $\rho$ and then applying renormalized Feynman rules viewed as taking values in motivic integrals on the left and applying cut Feynman rules giving de Rham periods on the right, see \cite{Algebraic, KlannMSc}.  Here we will stick to applying $\widetilde{\rho}$ but will go beyond pure diagrammatics by then using graph properties which can be derived from the Feynman rules in the Galois coaction to do further combinatorial simplification.  The idea is to lay out how far we can get with pure combinatorics before one would need to look into the details of the integrands.

In Chapter 4 of \cite{KlannMSc} Klann calculates the Galois coaction on the two loop photon self energy diagrams in QED and then takes an IR limit where the internal photon has very small energy, what is sometimes called a \emph{soft photon}.  He then notices that the result of the Galois coaction simplifies, with only one graph appearing on the left hand side of the tensor and the expression on the right being the diagrams contributing to the cancellation of IR singularities given by the Kinoshita-Lee-Nauenberg (KLN) theorem at two loops.  In this way he demonstrates the compatibility of the KLN theorem and the coaction at two loops.  The calculations are quite lengthy, in part because of the use of $\rho$ instead of $\widetilde{\rho}$, so let us return to this calculation now that we have $\widetilde{\rho}$ in hand, at least as far as the combinatorial considerations of the calculation.

Let us first consider the case where $G$ is a connected 1PI graph in QED with two external photon edges and no other external edges.  These are the connected 1PI graphs contributing to the photon self energy. We wish to understand $\widetilde{\rho}(G)$.  With no extra assumptions we cannot do any better than \eqref{eq cut coprod} for an expression for $\widetilde{\rho}(G)$, however there are two important consequences of applying Feynman rules that we can easily encorporate combinatorially.  The first consequence is that if any connected component after the cut has no external edges then the entire expression will be $0$.  This is due to the renormalization conditions that are used in this area (and is an argument for using kinematic-type renormalization schemes).  Combinatorially, this is an easy yet very useful condition to incorporate into $\widetilde{\rho}(G)$ -- instead of summing over all cuts, we only need to sum over cuts into two pieces for which one external edge is in each piece giving
\[
\sum_{\substack{C \text{ cut of } G \text{ into}\\ G_1 \text{ and  } G_2 \\ \text{separating external edges}}}\left(\sum_{\substack{T_i \text{ spanning}\\\text{tree of } G_i}} G/(T_1\cup T_2)\right) \otimes G_1\cup G_2
\]
The graphs on the left hand side are now of a quite particular form.  They all have exactly two vertices, one for each external edge.  There is a multiple edge consisting of at least two edges between the two vertices, and there may also be self-loops (tadpoles) at one or both of the vertices.

Another property of the Feynman rules is that fermion tadpoles vanish.  As discussed in Section~\ref{sec tree indep} this can be incorporated into the combinatorics or they can be removed as a separate step after applying $\widetilde{\rho}$.  In any case, as remarked at the end of Section~\ref{sec tree indep}, forbidding fermion tadpoles implies that all the cuts $C$ in the sum must cut all of the fermion loops of $G$.  The contraction of the spanning trees $T_1$ and $T_2$ also must not create fermion loops, and so all remaining internal fermions in $G_1$ and $G_2$ (which necessarily form forests since all fermion loops were cut) must be in $T_1$ and $T_2$.  

Suppose $G$ contains $k$ fermion loops and has loop number $\ell$.  The edges contracted to give the left hand sides in $\widetilde{\rho}$ are the edges not in a spanning forest of $G$ with two trees.  There are always $\ell+1$ such edges, so all graphs on the left hand sides in $\widetilde{\rho}$ have $\ell+1$ edges.  Of these edges, an even number at least $2k$ must be fermions going between the two vertices and the remaining are photons which may either be self-loops on either side or also go between the vertices.  The example in Figure~\ref{fig:QEDcut} shows how cuts with more that $2k$ fermions can appear; in that example $k=1$.

Next consider the effect of taking an IR limit where some of the photons are taken to be soft.  In this soft limit, a photon going between the two vertices of a left hand side graph does not transfer any momentum, and so it can be factored out.  Tadpoles, whether soft or not, also don't transfer momentum and so can factor out.  See Section 4.4 of \cite{KlannMSc} for the calculations demonstrating this at two loops.  Now more of the graphs on the left hand side in $\widetilde{\rho}(G)$ become equivalent as soft photons in this limit are the same whether they are appearing as tadpoles or as edges between the two vertices.  Graphically we can draw edges which factor out as separate connected components like in \cite{KlannMSc} equation 4.63.  

A final consideration is whether or not we can include soft photons in $T_1$ or $T_2$.  The edges of $T_1$ and $T_2$ are contracted on the left hand side, but contracting an edge is physically saying that it does not propagate and so it puts no constraint on how the momenta of its neighbours flow between themselves.  This is directly opposed to taking a soft limit, and so, similarly to how an edge cannot simultaneously be contracted and deleted, it is reasonable to impose that soft photons cannot be contracted and hence cannot appear in $T_1$ or $T_2$.  Every soft photon must, then, be on the left hand side of every term of $\widetilde{\rho}(G)$.

Note that there is a maximum number of photons that can appear in a left hand side graph, namely $\ell+1-2k$ since every fermion loop must be cut.  In the case that exactly $\ell+1-2k$ photons are soft, things are now particularly simple as all photons must appear on the left hand side of $\widetilde{\rho}(G)$ and no other photons can appear there, so all photons on the left hand side can be factored out, and further the number of fermion edges must be the same in each left hand graph since the number of photon edges is the same in each and the total number of edges is always the same in each.  With all of these considerations taken into account we get that $\widetilde{\rho}(G)$ has simplified to
\[
b_{2k} r^{\ell+1-2k} \otimes \sum_{\substack{C \text{ cut of } G \text{ into } G_1 \text{ and  } G_2 \\ \text{separating external edges}\\\text{cutting all fermion loops} \\\text{leaving no soft photon bridges}}} G_1\cup G_2
\]
where $b_{2k}$ is the graph with two vertices, one with each external edge, and $2k$ fermions between them, $r$ is the graph with one vertex and a soft photon loop, and bridge is used in the sense of graphs theory.

The situation in \cite{KlannMSc} was even nicer in that all internal photons were soft (there was only one internal photon in his situation).  Now if a QED graph with two external photons has loop order $\ell$ then it has $3\ell-1$ internal edges total, $2\ell$ of which are fermions and $\ell+1$ of which are photons.  So if we want to set the maximum number of photons to be soft and we want this maximum to be all the photons then we need $\ell+1-2k = \ell+1$ so we need $k=1$, that is a single fermion loop.  By the same calculations, in the case of a single fermion loop, setting the maximum number of photons to be soft is equivalent to setting all photons to be soft.  The case with a single fermion loop has a name and a great deal of literature about it; it is the case of \emph{quenched} QED.  Quenched QED with all internal photons soft then is the nicest case in this formulation as the action of $\widetilde{\rho}$ becomes
\[
b_{2} r^{\ell-1} \otimes \sum_{C} G_1\cup G_2
\]
where the cuts $C$ are those induced by cutting the fermion loop so that one external edge is on each side.

\medskip

We can make a similar combinatorial analysis when there are more than two external edges.  Cuts still must have at least one external edge in each component, all fermion loops must still be cut, tadpoles (necessarily photon tadpoles) and soft internal photons all factor out on the left hand side, and the remaining graph has one vertex for each component of the cut with fermion and non-soft photons between the vertices.  If we additionally forbid contracting soft photons then we know all soft photons will appear (and can be factored out) on the left hand side.  In the case with the maximum number of soft photons, different graphs are now possible on the left, but they remain fairly tame with only a limited ability to redistribute the fermion edges between the vertices.

\bibliographystyle{amsplain}
\bibliography{phi_4.bib}% Produces the bibliography via BibTeX.

\end{document}